\documentclass[sigconf]{acmart}
\AtBeginDocument{%
  }

\setcopyright{acmlicensed}
\copyrightyear{2026}
\acmYear{2026}
\acmDOI{XXXXXXX.XXXXXXX}
\acmConference[Conference acronym 'XX]{Make sure to enter the correct
  conference title from your rights confirmation email}{June 03--05,
  2018}{Woodstock, NY}
\acmISBN{978-1-4503-XXXX-X/2018/06}

\usepackage{tikz}
\usepackage{cleveref}
\usetikzlibrary{shapes.geometric, positioning, fit, arrows}
\usepackage{todonotes}
\newcommand{\pwline}[1]{\todo[author=Przemek,inline,backgroundcolor=green!20]{#1}}

\usepackage{thm-restate}
\declaretheorem[name=Theorem]{thm}
\newtheorem{theorem}[thm]{Theorem}

\newtheorem*{example*}{Example}
\newtheorem{definition}[thm]{Definition}
\newtheorem{lemma}[thm]{Lemma}
\newtheorem{proposition}[thm]{Proposition}
\newtheorem{corollary}[thm]{Corollary}
\newtheorem*{claim*}{Claim}
\newtheorem*{theorem*}{Theorem}

\providecommand{\lBrace}{\mathopen{\lbrace\mkern-3mu\vert}}
\providecommand{\rBrace}{\mathclose{\vert\mkern-3mu\rbrace}}

\newcommand{\true}{\ensuremath{\mathsf{true}}}
\newcommand{\false}{\ensuremath{\mathsf{false}}}
\newcommand{\prop}{\ensuremath{\mathsf{PROP}}}

\newcommand{\M}{\ensuremath{\mathfrak{M}}}
\newcommand{\N}{\ensuremath{\mathfrak{N}}}

\newcommand{\FO}{\ensuremath{\text{FO}}}

\newcommand{\ML}{\ensuremath{\mathcal{ML}}}

\newcommand{\EML}{\ensuremath{\exists\mathcal{ML}}}
\newcommand{\EPML}{\ensuremath{\exists^+\mathcal{ML}}}

\newcommand{\GML}{\ensuremath{\mathcal{GML}}}
\newcommand{\EGML}{\ensuremath{\exists\mathcal{GML}}}
\newcommand{\EPGML}{\ensuremath{\exists^+\mathcal{GML}}}

\newcommand{\ALCQ}{\ensuremath{\mathcal{ALCQ}}}
\newcommand{\ELUQ}{\ensuremath{\mathcal{ELUQ}}}

\newcommand{\G}{\ensuremath{G}}

\newcommand{\emb}{\ensuremath{\lambda}}
\newcommand{\unr}{\ensuremath{\mathsf{Unr}}}

\newcommand{\GN}{\ensuremath{\mathcal{N}}}
\newcommand{\C}{\ensuremath{\mathcal{C}}}
\newcommand{\T}{\ensuremath{\mathcal{T}}}
\newcommand{\Lr}[1][]{%
  \ifx#1\empty
    \ensuremath{\mathcal{L}}%
  \else
    \ensuremath{\mathcal{L}^{#1}}%
  \fi
}   
\newcommand{\comb}{\ensuremath{\mathsf{comb}}}
\newcommand{\agg}{\ensuremath{\mathsf{agg}}}

\newcommand{\cls}{\ensuremath{\mathsf{cls}}}

\newcommand{\Tmin}{\ensuremath{\mathcal{T}_{\mathsf{min}}}}

\begin{document}

\title{Preservation Theorems for Unravelling-Invariant Classes: A Uniform Approach for Modal Logics and Graph Neural Networks}

\renewcommand{\shorttitle}{Preservation Theorems for Unravelling-Invariant Classes}


% \author{Anonymous}
% \affiliation{%
%   \institution{Anonymous}
%   \city{Anonymous}
%   \country{Anonymous}}
% \email{Anonymous}

\author{Przemys\l aw Andrzej Wa\l \k{e}ga}
\affiliation{%
  \institution{Queen Mary University of London}
  \city{London}
  \country{UK}}
\email{p.walega@qmul.ac.uk}

\author{Bernardo Cuenca Grau}
\affiliation{%
  \institution{University of Oxford}
  \city{Oxford}
  \country{UK}}
\email{bernardo.grau@cs.ox.ac.uk}

%\renewcommand{\shortauthors}{Anonymous}

\begin{abstract}
We study preservation theorems for modal logics over finite structures with respect to three fundamental semantic relations: embeddings, injective homomorphisms, and homomorphisms. We focus on classes of pointed Kripke models that are invariant under bounded unravellings, a natural locality condition satisfied by modal logics and by graph neural networks (GNNs). We show that preservation under embeddings coincides with definability in existential graded modal logic; preservation under injective homomorphisms with definability in existential positive graded modal logic; and preservation under homomorphisms with definability in existential positive modal logic. A key technical contribution is a structural well-quasi-ordering result. We prove that the embedding relation on classes of  tree-shaped models of uniformly bounded height forms a well-quasi-order, and that the bounded-height assumption is essential. This well-quasi-ordering yields a finite minimal-tree argument leading to explicit syntactic characterisations via finite disjunctions of (graded) modal formulae. 

As an application, we derive consequences for the expressive power of GNNs. Using our preservation theorem for injective homomorphisms, we obtain a new logical characterisation of monotonic GNNs, showing that they capture exactly existential-positive graded modal logic, while monotonic GNNs with MAX aggregation correspond precisely to existential-positive modal logic.
\end{abstract}



\received{20 February 2007}
\received[revised]{12 March 2009}
\received[accepted]{5 June 2009}

\maketitle

\section{Introduction}

Preservation theorems are a central theme in model theory
and theoretical computer science. They characterise
closure properties of classes of structures---such as preservation under embeddings or homomorphisms---with definability in specific logical fragments.
Łoś-Tarski, Lyndon, and Homomorphism
Preservation theorems are three classical examples \cite{los1955extending,lyndon1959properties,DBLP:books/daglib/0030198,tarski1954contributions}. Each  states that
a particular fragment of first-order logic---existential,
positive, and existential-positive, respectively---is, up to logical equivalence, exactly as expressive as the class of first-order formulae whose models are preserved under a corresponding relation between structures: embeddings, surjective homomorphisms, and homomorphisms. 

While these classical preservation theorems hold over arbitrary structures, many of them fail when restricted to finite structures. 
Standard proofs rely on compactness and other infinitary techniques
from classical model theory \cite{tent2012course}, which are unavailable in the finite setting \cite{DBLP:books/sp/Libkin04,DBLP:journals/bsl/Rosen02,DBLP:series/txtcs/GradelKLMSVVW07}. Only a limited number of classical tools, most notably  Ehrenfeucht-Fraïssé games, are applicable also to  finite model theory,  but they are insufficient to
bridge the gap left by the failure of compactness. As a consequence, important preservation theorems such as  Łoś-Tarski and Lyndon are known to fail over finite structures  \cite{DBLP:journals/bsl/Rosen02, DBLP:books/sp/Libkin04}. By contrast, in a breakthrough result, Rossman  showed that  the Homomorphism Preservation Theorem does remain valid over finite structures \cite{DBLP:journals/jacm/Rossman08}.

Preservation theorems have also been extensively investigated in the setting of modal logics. Unlike first-order logic, standard modal logics are
inherently local and enjoy the tree-model property \cite{DBLP:journals/jphil/AndrekaNB98,DBLP:books/el/07/BBW2007,DBLP:conf/wollic/ArecesHD10}. 
Modal semantics is invariant under bisimulation and admits canonical tree representations via unravelling.
Van Benthem's classical characterisation theorem shows that, over arbitrary structures, first-order properties invariant under bisimulation are exactly those definable in the basic modal logic $\ML$. Rosen  subsequently showed that this correspondence continues to hold over finite structures, establishing that bisimulation invariance suffices for modal definability even in the absence of compactness \cite{DBLP:journals/jolli/Rosen97,DBLP:journals/bsl/Rosen02}.

%\begin{table*}[t] 
%\centering
%\setlength{\tabcolsep}{4pt}
%\caption{New preservation theorems for (graded) modal logic, and known results which follow from our results
%}
%\label{results}
%\begin{tabular}{l|l|l|l}
%\toprule
% & \multicolumn{3}{c}{Preserved under:} 
% \\
%\cmidrule(lr){2-4}
%& embeddings 
% & injective homomorphisms 
% & homomorphisms 
% \\
%\midrule
%\GML & \EGML{} (\Cref{embeddingGML})& \EPGML{} (\Cref{injectiveGML})   & \EPML{} \cite{DBLP:journals/apal/AbramskyR24} 
%\\
%\ML  & \EML{} %(\Cref{embeddingML}) 
%\cite{DBLP:journals/jolli/Rosen97} & \EPML{} (\Cref{injhomML}) & \EPML{} \cite{DBLP:journals/apal/AbramskyR24}  
%\\
%\bottomrule
%\end{tabular}
%\end{table*}

Furthermore,  \citet{DBLP:journals/jolli/Rosen97} proved a finite modal analogue to  Łoś-Tarski's preservation theorem
 showing that, over finite pointed models, $\ML$-formulae
 preserved under embeddings are definable in the existential fragment of $\ML$.  More recently, 
\citet{DBLP:journals/apal/AbramskyR24}  developed 
a unifying framework for homomorphism preservation theorems exploiting 
category-theoretic methods and game comonads \cite{DBLP:journals/fuin/Abramsky22,DBLP:conf/lics/AbramskyDW17}. In particular, they show that for $\ML$, as well as for graded modal logic $\GML$,  preservation under homomorphisms is equivalent to definability in the existential-positive fragment, with no increase in modal depth. Their results apply  to finite and infinite structures without reliance on compactness.

In this paper, we study  modal preservation theorems from a
finite model-theoretic perspective. We work with classes of finite pointed Kripke models that are invariant under bounded unravellings, a natural locality condition in modal logic. Within this setting, we establish preservation theorems for embeddings, injective homomorphisms, and homomorphisms, obtaining exact characterisations in terms of existential and existential-positive fragments of $\GML$ and $\ML$. Our results complement Rosen's finite preservation theorems and the categorical framework of Abramsky and Reggio by providing  syntactic characterisations. In particular, we obtain new preservation theorems for embeddings and injective homomorphisms  in the setting of $\GML$. Our results imply also the results over finite models by  Rosen as well as Abramsky and Reggio. 

An important technical ingredient underlying our approach is a structural result concerning well-quasi-orders (wqos). In \Cref{sec:wqo} we show that the embedding relation on classes of tree-shaped models of  uniformly bounded height forms a well-quasi-order (\Cref{embedding_wqo}), and that this bounded-height assumption is essential: without it, none of the considered relations is guaranteed to yield a wqo (\Cref{unbounded_bad}). Our use of wqos is closely related in spirit to the framework 
developed by  \citet{DBLP:phd/hal/Lopez23a}, which studies finite preservation theorems via minimal models, wqos and Noetherian topologies. 
The two approaches, however, address different settings: Lopez's results rely on hereditary classes of finite structures, an assumption that fails for unravelling-invariant classes of Kripke models (see \Cref{non_closed}), and hence require different techniques.

In \Cref{sec:pres}, building on our  results for wqos, we establish general preservation theorems for classes $\C$ of finite pointed Kripke models invariant under bounded unravellings. From any such class $\C$ preserved under one of the semantic relations discussed above, we extract a finite set of minimal tree-shaped models. Each such minimal tree is definable by a modal formula of bounded depth and the finite disjunction of these formulae  yields a defining formula for $\C$.  
Concretely, we show that 
$\C$ is preserved under embeddings if and only if it is definable in existential  \GML{} (\Cref{th:preserve}); that $\C$ is preserved under injective homomorphisms if and only if it is definable in existential-positive \GML{} 
(\Cref{th:exists-positive-unravelling}); and that  $\C$ is preserved under  homomorphisms if and only if it is definable in existential-positive \ML{} (\Cref{th:exists-positive-unravelling-ML}). 
From these results, we derive corresponding preservation theorems for \GML{} and \ML{} formulae, summarised in \Cref{results}.

As we show in \Cref{sec:gnn}, beyond their intrinsic interest in logic, our preservation theorems have direct implications for the expressive power of graph neural networks (GNNs) \cite{DBLP:conf/icml/GilmerSRVD17}---machine learning models for graph data---by enforcing structural invariants such as node permutation invariance (i.e.\ invariance under graph isomorphism) \cite{DBLP:series/synthesis/2020Hamilton,DBLP:journals/corr/abs-2104-13478,DBLP:journals/corr/abs-2507-18145}.
In their most common form, GNNs are message-passing architectures that compute node representations by iteratively aggregating information from neighbouring nodes. Their expressive power---the classes of graph properties they can represent---is thus tied to logical notions of locality and unravelling \cite{DBLP:conf/iclr/BarceloKM0RS20}. A growing body of work has established  correspondences between classes of GNN architectures and logic, including basic and graded modal logic, two-variable fragments, and Datalog \cite{DBLP:conf/iclr/BarceloKM0RS20,grohe2024descriptive,DBLP:journals/corr/abs-2508-06091,DBLP:conf/kr/CucalaGMK23,DBLP:conf/kr/CucalaG24,DBLP:conf/ijcai/NunnSST24,benedikt_et_al:LIPIcs.ICALP.2024.127,tena2025expressive,boundedGNNs,DBLP:conf/nips/AhvonenHKL24}. In particular, monotonic GNNs---architectures in which  all  activation and aggregation functions are monotonic and matrix weights are  non-negative---are representable by non-recursive tree-shaped Datalog programs with inequalities, whereas
monotonic GNNs using only MAX aggregation  can be represented without inequalities \cite{DBLP:conf/kr/CucalaGMK23,tena2025expressive}.

Our results yield a new logical characterisation of monotonic GNNs. Using the preservation theorem for injective homomorphisms, we show that monotonic GNNs have exactly the expressive power of existential-positive \GML{}  whereas monotonic-MAX GNNs have the expressive power of existential-positive \ML{}. These results provide a uniform account grounded in modal preservation theory.
More broadly, our approach illustrates how  tools from finite model theory, such as unravelling-invariance and wqos, can be fruitfully applied to questions in the theoretical analysis of machine learning models.

\begin{table}[t]
\centering
\small
\setlength{\tabcolsep}{4pt}
\caption{Preservation results for (graded) modal logic.}
\label{results}
\begin{tabular}{l l l l}
\toprule
 & Embeddings & Injective homomor. & Homomor. \\
\midrule
\GML & \EGML (\Cref{embeddingGML}) & \EPGML (\Cref{injectiveGML}) & \EPML \cite{DBLP:journals/apal/AbramskyR24} \\
\ML  & \EML \cite{DBLP:journals/jolli/Rosen97} & \EPML (\Cref{injhomML}) & \EPML \cite{DBLP:journals/apal/AbramskyR24} \\
\bottomrule
\end{tabular}
\end{table}

%%%%%%%%%%%%%%%%%%%
\section{Preliminaries}
%%%%%%%%%%%%%%%%%%%%%

\noindent\textbf{Kripke Models.}
All models considered in this paper are  finite and given in the modal vocabulary over a finite set
$\prop$ of propositional symbols.
Equivalently, they can be viewed as finite first-order $\sigma$-structures, where the vocabulary
$\sigma$ contains a single binary relation symbol $R$
and, for each proposition $p_i \in \prop$, a unary predicate symbol $P_i$.

We  use standard Kripke semantics \cite{DBLP:books/el/07/BBW2007}: a model $\M$ is a triple $(W,R,V)$, where $W$ is a non-empty finite  set of \emph{worlds}, $R \subseteq W \times W$ is the \emph{accessibility relation}, and $V$ is a \emph{valuation function} mapping each $p_i \in \prop$ to a subset of $W$. 
A \emph{pointed model} is a pair $(\M,w)$ with $w \in W$.
A pointed model $(\M,w)$  is \emph{tree-shaped}, 
if  for every world $v \in W$ there is a unique directed path from $w$ to $v$; we call $w$ the \emph{root}.
The height of a tree-shaped model is the length of its longest directed path from the root.

\smallskip
\noindent\textbf{Semantic Relations.}
Let $(\M, w)$ and $(\N, v)$ be pointed models, with
$\M = (W,R,V)$ and $\N = (W',R',V')$. 
We consider functions $f: W \rightarrow W'$ such that $f(w) = v$.

A function $f$ is an \emph{isomorphism} between  $(\M, w)$ and $(\N, v)$ if it is a bijection and, 
for all $u,z \in W$ and each $p_i\in \prop$ it satisfies $u R z$  iff $f(u) R' f(z)$, and $u \in V(p_i)$ iff $f(u) \in V'(p_i)$.

A function $f$ is an \emph{embedding} from $(\M, w)$ to $(\N, v)$ if it is an isomorphism between $(\M, w)$ and the submodel
 of $(\N, v)$  induced by the image $f(W)$.

A function $f$ is a \emph{homomorphism} from $(\M, w)$ to $(\N, v)$ if, for all $u,z \in W$ and each $p_i \in \prop$, $u R z$ implies $f(u) R' f(z)$, and $u \in V(p_i)$ implies $f(u) \in V'(p_i)$. 
 If $f$ is injective, we call it an \emph{injective homomorphism}.

Clearly, every isomorphism is an embedding, every embedding is an 
injective homomorphism, and every injective homomorphism
is a homomorphism.

\smallskip
\noindent\textbf{Unravelling.}
Let $(\M, w)$ be a pointed model with $\M = (W,R,V)$ and let $L \in \mathbb{N}$.
The $L$-\emph{unravelling} of $(\M, w)$, denoted
$\unr^L(\M,w)$, is the pointed model $((W', R', V'), w')$ defined as follows:
\begin{itemize}
\item $W'$ contains a world \((w, v_1, \ldots, v_{\ell})\) for each path $wRv_1$, $v_1Rv_2$, $\ldots$, $v_{\ell-1}Rv_{\ell}$ of length $\ell \leq L$ in $\M$.
\item The distinguished world $w'$ is the world corresponding to the  path $(w)$ of length $0$.
\item We have  \((w, v_1, \ldots, v_{\ell-1})\) $R'$  \((w, v_1, \ldots, v_{\ell})\)  iff $v_{\ell-1}Rv_{\ell}$.
\item For each $p_i \in \prop$, we define
$$V'(p_i) = \{ (w, v_1, \ldots, v_{\ell}) \in W' \mid v_{\ell} \in V(p_i)\}$$.
\end{itemize}
An example of a model and its $3$-unravelling is depicted in \Cref{fig::unravelling}.
An $L$-unravelling is  a tree-shaped model of height at most $L$.

\begin{figure}[ht]
\centering
\begin{tikzpicture}
\tikzset{
>=latex,
node distance=1.0cm,
world/.style={
    draw,
    circle,
    fill,
    minimum size=1.5mm,
    inner sep=0pt,
    scale=1
  }
}

\begin{scope}[local bounding box=M1]
% ----- nodes -----
\node[world, label=right:] (v1) at (0,0) {};
\node[world, below left of=v1, label=right:$p_1$] (v2)  {};
\node[world, below of=v1, label=right:$p_2$] (v3)  {};
\node[world, below right of=v1, label=right:$p_1$] (v4)  {};
% ----- edges -----
\draw[->] (v1) -- (v2);
\draw[->] (v1) -- (v4);
\draw[->] (v1) to[bend left=20] (v3);
\draw[->] (v3) to[bend left=20] (v1);
\end{scope}
% ----- frame around the scope -----
\node[draw, rounded corners, fit=(M1), inner sep=6pt] (frame1) {};
% ----- label above -----
\node[above=2pt of frame1] {$(\M,w)$};

\begin{scope}[xshift=4cm, local bounding box=M2]
% ----- nodes -----
\node[world, label=right:] (v1) at (0,0) {};
\node[world, below left of=v1, label=right:$p_1$] (v2)  {};
\node[world, below of=v1, label=right:$p_2$] (v3)  {};
\node[world, below right of=v1, label=right:$p_1$] (v4)  {};

\node[world, below of=v3, label=right:] (u)  {};

\node[world, below left of=u, label=right:$p_1$] (w2)  {};
\node[world, below of=u, label=right:$p_2$] (w3)  {};
\node[world, below right of=u, label=right:$p_1$] (w4)  {};
% ----- edges -----
\draw[->] (v1) -- (v2);
\draw[->] (v1) -- (v4);
\draw[->] (v1) -- (v3);

\draw[->] (v3) -- (u);

\draw[->] (u) -- (w2);
\draw[->] (u) -- (w4);
\draw[->] (u) -- (w3);
\end{scope}
% ----- frame around the scope -----
\node[draw, rounded corners, fit=(M2), inner sep=6pt] (frame2) {};
% ----- label above -----
\node[above=2pt of frame2] {$\unr^3(\M,w)$};
%\node[above=15pt of frame2] {$3$-unravelling};

\end{tikzpicture}
\caption{A model and its $3$-unravelling}
\label{fig::unravelling}
\end{figure}
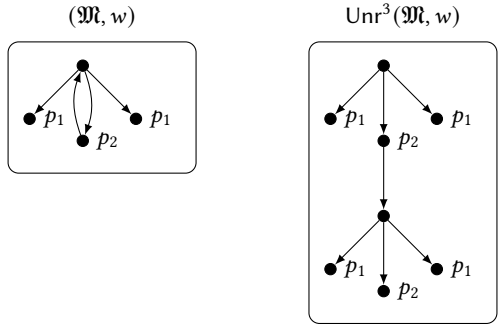

\smallskip
\noindent\textbf{Bisimulation.}
Since we work throughout with bounded unravellings, we introduce a notion of bisimulation parametrised by depth.
Let $(\M,w)$ and $(\N,v)$ be pointed Kripke models, where
$\M = (W,R,V)$ and $\N = (W',R',V')$, and let $L \in \mathbb{N}$. 
An \emph{$L$-bisimulation} between  $(\M,w)$ and $(\N,v)$ is a relation $Z \subseteq W \times W'$ such that
$w Z v$ and, for all $uZu'$ the following conditions hold:
\begin{itemize}
\item for each $p_i \in \prop$, $u \in V(p_i)$ iff  $u' \in V'(p_i)$;
\item if $u R z$ and the
distance from $u$ to the root $w$ is less than $L$, there exists
$z' \in W'$ such that $u' R' z'$ and $z Z z'$. 
\item if $u' R' z'$  and the distance from $u'$ to the root $v$ is less than $L$, there exists
$z \in W$ such that $u R z$ and $z Z z'$.
\end{itemize}
If such a relation exists, $(\M,w)$ and $(\N,v)$ are $L$-bisimilar and we write  $(\M,w) \sim_{L} (\N,v)$.
For each pointed model $(\M,w)$ and $L \in \mathbb{N}$, $(\M,w)$ is $L$-bisimilar to its $L$-unravelling $\unr^L(\M,w)$.

\smallskip
\noindent\textbf{Preservation and Invariance.}
Let  $\C$ be a class of pointed models and $\preceq$  a binary
relation on pointed models.
\begin{itemize}
\item\emph{Preservation.} Class $\C$ is \emph{preserved} under $\preceq$ if $(\M,w) \in \C$ and $(\M,w) \preceq (\N,v)$ imply that $(\N,v) \in \C$.
\item \emph{Invariance.} Class  $\C$ is \emph{invariant} under $\preceq$ if it is preserved under both $\preceq$ and its inverse relation $\preceq^{-1}$.
\end{itemize}
We consider only classes of pointed models that are invariant under isomorphisms, and we freely identify
models with their isomorphism type.  Moreover, we study
classes preserved under embedding, injective homomorphism, and homomorphism.
Each of these notions induces a corresponding relation $\preceq$ on pointed models, where
$(\M,w) \preceq (\N,v)$ holds if there exists a map of the given type from $(\M,w)$ to $(\N,v)$.
Characterising classes preserved under these relations is a central topic in finite model theory.

Regarding invariance, we focus on classes of pointed models that are invariant 
under $L$-unravellings; that is classes $\C$ such that  $(\M,w) \in \C$ if and only if $\unr^L(\M,w) \in \C$.
Classes invariant under unravelling have been extensively studied in the literature.  Crucially for our results, both modal logic formulae and graph neural networks induce only classes of models that are invariant under unravelling, a fact that we exploit throughout the paper.

\smallskip
\noindent\textbf{Well-quasi-orders.}
Our results exploit well-quasi-orders (wqos) and their standard properties, which we recall. Let $A$ be a  set  and let $\preceq$ be a binary relation  on $A$.
The pair $(A,\preceq)$  is a \emph{well-quasi-order}  if it is a quasi-order (that is,  reflexive and transitive) and each infinite sequence $a_1,a_2, \dots$
of elements 
of $A$ contains indices $i < j$ such that $a_i \preceq a_j$.
Equivalently, $\preceq$ is a wqo if it is well-founded and has no infinite antichains.
In particular, in a wqo every subset has  finitely many minimal elements, since an infinite set of minimal elements forms an infinite antichain.

We will use the following basic properties of wqos. If $(A,\preceq)$ is a  wqo, then $(A',\preceq)$, with $A' \subseteq A$ is also a wqo. 
Moreover,
if $\preceq'$ is a quasi-order on $A$ extending $\preceq$, then $(A,\preceq')$ is again a wqo.

Another key property of wqos used in this paper is Higman's Lemma. 
Let $A^{*}$ be the set of all finite sequences  with elements in $A$ and 
let $\preceq^{*}$ be the \emph{subsequence order}, defined by $(a_1,\ldots,a_n) \preceq^{*} (b_1,\ldots,b_m)$ if there exists a strictly increasing map 
$f:\{1,\ldots,n\} \to \{1,\ldots,m\}$ such that $a_i \preceq b_{f(i)}$ for all $i \in \{1,\ldots,n\}$.

\begin{lemma}[Higman's Lemma]\label{higman}
If $(A,\preceq)$ is a wqo, then $(A^{*},\preceq^{*})$ is also a wqo.
\end{lemma}

\smallskip
\noindent\textbf{Modal Logics.}
Our preservation theorems characterise classes of models in terms of their definability in modal logics.
In particular, we consider \emph{graded modal logic} (\GML), \emph{basic modal logic} (\ML), and several of their fragments.
We briefly recall their syntax and semantics.
Throughout the paper, we fix a finite signature  $\prop = \{p_1,\ldots, p_n\}$ of propositional symbols. 
Formulae of $\GML$ \cite{DBLP:journals/sLogica/Rijke00,DBLP:journals/logcom/Tobies01,DBLP:journals/corr/abs-1910-00039}  over $\prop$ are defined inductively by the grammar
$$
\varphi :=  p \mid \neg \varphi \mid \varphi \land \varphi \mid \varphi \lor \varphi \mid \Diamond^{\geq k} \varphi,
$$
where $p \in \prop$ and, for each $k \geq 1$,  $\Diamond^{\geq k}$
is the $k$-graded modal operator.
The (modal) \emph{depth} of a \GML-formula $\varphi$ is the maximum nesting of graded modal operators in $\varphi$. 

The \emph{existential} fragment of $\GML$, denoted $\EGML$, is obtained  by allowing negation only in front of propositional symbols. The \emph{existential-positive} fragment, denoted $\EPGML$, further restricts formulae to be negation-free.
The logic \ML{} is obtained by restricting the  graded modal operators to $k=1$, and writing $\Diamond$ for $\Diamond^{\geq 1}$.
The fragments \EML{} and \EPML{} of \ML{} are defined analogously.

Formulae are interpreted over Kripke models $\M=(W,R,V)$.
The \emph{satisfaction relation}  is defined inductively as follows:
\begin{align*}
& (\M,  w)  \models  p  && \text{iff} &&   w \in V(p), \text{ where }  p \in \prop
\\
& (\M,  w)  \models  \neg  \varphi  && \text{iff} &&  (\M,  w)   \not \models   \varphi
\\
& (\M,  w)  \models \varphi \land \psi  && \text{iff} && (\M,  w)  \models   \varphi \text{ and } (\M,  w)  \models    \psi
\\
& (\M,  w)  \models \varphi \lor \psi  && \text{iff} && (\M,  w)  \models   \varphi \text{ or } (\M,  w)  \models    \psi
\\
& (\M,  w)  \models \Diamond^{\geq k}  \varphi && \text{iff} && \text{there exist at least $k$ worlds } v \in W 
\\
& && && \text{ such that }  wRv \text{ and }  (\M,  v)  \models   \varphi
\end{align*}
Given pointed models $(\M,w)$ and $(\N,v)$, a modal logic $\mathcal{L}$, and $\ell \geq 0$,  we write $(\M,w) \equiv_{\ell}^{\mathcal{L}}(\N,v)$ if both pointed models satisfy exactly the same $\mathcal{L}$-formulae of depth at most $\ell$.
For $\ML$, it holds that $(\M,w) \equiv_{\ell}^{\ML}(\N,v)$ if and only if  $(\M,w) \sim_{\ell}(\N,v)$, assuming
models are finite.
It is also well known that for each  $\GML{}$-formula $\varphi$ of depth $L$, and each pointed model $(\M,w)$, we have $(\M,w) \models \varphi$ if and only if $\unr^L(\M,w) \models \varphi$.
That is, the class of pointed models of $\varphi$ is  invariant under $L$-unravelling. 
A class $\mathcal{C}$ of pointed models is \emph{definable} in  a modal logic $\mathcal{L}$ if there exists an $\mathcal{L}$-formula $\varphi$ such that
$(\M,v) \models \varphi$ if and only if $(\M,v)  \in \mathcal{C}$.

%%%%%%%%%%
\section{Well-quasi-orders of Tree-Shaped Models of Bounded Height}\label{sec:wqo}
%%%%%%%%%%%

In this section we establish  structural results that form the technical core of the paper. We study the embedding, injective homomorphism, and homomorphism relations on classes of tree-shaped pointed Kripke models, with particular attention to the role played by bounded height. 
Throughout the paper, we work in finite model theory; accordingly, all models considered are finite.

We begin by highlighting the general strategy underlying our preservation results. Let $\preceq$ denote one of the embedding, injective homomorphism, and homomorphism relations.
Given a class $\C$  of pointed models preserved under $\preceq$,  we
fix a bound $L$ and consider the class $\T$ of  $L$-unravellings
of models in $\C$. 
By construction, $\T$ consists of pointed tree-shaped models of height at most $L$. We then consider the set
$\Tmin \subseteq \T$ of $\preceq$-minimal pointed models in $\T$.

For each pointed model in  $\Tmin$, we construct a modal formula that characterises it. If $\Tmin$ is finite, the disjunction of these formulae defines the class $\C$. Thus, the success of this
approach hinges on the finiteness of $\Tmin$. 

The main goal of this section is to justify this finiteness by establishing a stronger structural property, namely that
$\preceq$ is a well-quasi-order on $\T$. Since each wqo admits only finitely many minimal elements, this  implies that $\Tmin$ is finite. 
We thus proceed by establishing wqo results for the relations $\preceq$ over classes of tree-shaped models, starting with the case where $\preceq$ is the embedding relation.

We begin by showing that preservation under embeddings does not, in general, transfer to $L$-unravellings. In particular, even if a class of pointed models is preserved under embeddings, the corresponding
class of $L$-unravellings does not need to.

\begin{proposition}\label{non_closed}
For any $L > 1$, there exists a class of pointed models 
preserved under embeddings, but whose class of $L$-unravellings is not.
\end{proposition}

\begin{proof}
Fix $L >1$ and let  $(\M,w)$ be a pointed model
with $\M=(W,R,V)$ such that 
$W = \{w\}$, $R = \{ (w,w) \}$, and $V(p) = \emptyset$ for every proposition $p$, as depicted on the left of \Cref{notpreserved}.
Let $\C$
be the class of pointed models $(\N, v)$ whose
restriction  to $v$ is isomorphic to $(\M,w)$.
Equivalently, $\N$ has a self-loop at $v$ and no proposition holds at $v$. 
It follows that
$\C$ is preserved under embeddings.

Let $\T = \{\unr^L(\M',w') \mid (\M',w') \in \C \}$ be the class of all $L$-unravellings of models in $\C$.
In particular, the $L$-unravelling depicted in the middle of \Cref{notpreserved}, belongs to $\T$.
We will show that $\T$ is not preserved under embeddings. 

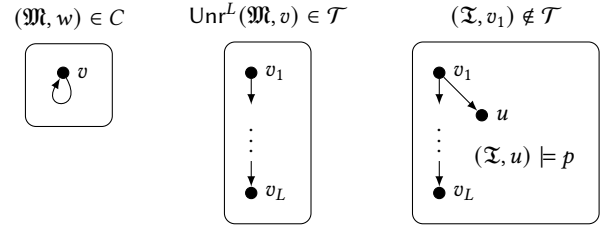
\begin{figure}[ht]
\centering
\begin{tikzpicture}
\tikzset{
>=latex,
node distance=0.8cm,
world/.style={
    draw,
    circle,
    fill,
    minimum size=1.5mm,
    inner sep=0pt,
    scale=1
  }
}

\begin{scope}[local bounding box=M1]
% ----- nodes -----
\node[world, label=right:$v$] (v) at (0,0) {};
% ----- edges -----
\path (v) edge[->, loop below, out=300, in=240, min distance=5mm] (v);
\end{scope}
% ----- frame around the scope -----
\node[draw, rounded corners, fit=(M1), inner sep=6pt] (frame1) {};
% ----- label above -----
\node[above=2pt of frame1] {$(\M,w) \in C$};

\begin{scope}[xshift=2.5cm, local bounding box=M2]
% ----- nodes -----
\node[world, label=right:$v_1$] (v1) at (0,0) {};
\node[below of=v1] (dots) {$\vdots$};
\node[world, below of=dots, label=right:$v_L$] (vL) {};
% ----- edges -----
\draw[->] (v1) -- (dots);
\draw[->] (dots) -- (vL);
\end{scope}
% ----- frame around the scope -----
\node[draw, rounded corners, fit=(M2), inner sep=6pt] (frame2) {};
% ----- label above -----
\node[above=2pt of frame2] {$\unr^L(\M,v) \in \T$};
%\node[above=15pt of frame2] {$L$-unravelling};

\begin{scope}[xshift=5cm, local bounding box=M3]
% ----- nodes -----
\node[world, label=right:$v_1$] (v1) at (0,0) {};
\node[below of=v1] (dots) {$\vdots$};
\node[world, below of=dots, label=right:$v_L$] (vL) {};
\node[world, below right of=v1, label=right:$u$] (u)  {};

\node[draw=none, below right of=u] {$(\mathfrak T,u)\models p$};

% ----- edges -----
\draw[->] (v1) -- (dots);
\draw[->] (dots) -- (vL);
\draw[->] (v1) -- (u);
\end{scope}
% ----- frame around the scope -----
\node[draw, rounded corners, fit=(M3), inner sep=6pt] (frame3) {};
% ----- label above -----
\node[above=2pt of frame3] {$(\mathfrak T,v_1) \notin \T$};

\end{tikzpicture}
\caption{Models $(\M,w)$, $\unr^L(\M,v)$, and  $(\mathfrak T,v_1)$}
\label{notpreserved}
\end{figure}

Consider the pointed tree-shaped model $(\mathfrak T,v_1)$ depicted on the right of \Cref{notpreserved}.
It consists of a directed path $v_1,\dots,v_L$ of length $L$ together with
an additional world $u$ such that $(v_1,u)$. A proposition $p$ holds at $u$  and no other proposition holds at any other world.
Clearly, $(\mathfrak T,v_1)$ is tree-shaped   and $\unr^L(\M,v)$ embeds into  $(\mathfrak T,v_1)$.

We claim that $(\mathfrak T,v_1)\notin \T$.
Suppose towards a contradiction that there exists $(\M',w') \in \C$ such that $\unr^L(\M',w')$ is isomorphic to $(\mathfrak T,v_1)$. 
Since $(\mathfrak T,v_1)$ has a neighbour $u$ of the root  satisfying $p$, while $w'$ itself does not satisfy $p$ (as $(\M', v') \in \C$), it follows that
$\M'$  must contain a  world reachable from $w'$ in one step that satisfies $p$.
Moreover, since $(\M', v') \in \C$, $\M'$ must contain a self-loop $(w',w')$.

As a consequence, $\unr^L(\M',w')$ 
contains distinct paths of lengths $1, \dots, L$ from the root to distinct nodes satisfying $p$. In particular,  
 $\unr^L(\M',w')$ contains 
$L$ distinct nodes satisfying $p$. By contrast,
$(\mathfrak T,v_1)$ contains exactly one world satisfying $p$.  
Since $L>1$, we obtain that $\unr^L(\M',w')$ cannot be isomorphic to $(\mathfrak T,v_1)$, contradicting
our assumption.
\end{proof}

Therefore, we cannot assume that the classes $\T$ of unravellings are preserved under embeddings.
To overcome this obstacle and ensure the finiteness of $\Tmin$,
we establish a stronger structural result.
Specifically, we show that the embedding relation forms a well-quasi-order on any class of tree-shaped models of uniformly bounded height.
 Since  $L$-unravellings have height at most $L$, this result applies directly to the
classes $\T$ considered above. 

This theorem, however, goes beyond the setting of
unravellings: it applies to arbitrary classes
of tree-shaped models, independently of any preservation properties.
In particular, it complements the framework developed by Lopez \cite{DBLP:phd/hal/Lopez23a}, which relies on hereditary classes of
finite structures. As shown in Proposition~\ref{non_closed}, unravelling-invariant classes of pointed Kripke models are not
hereditary and hence the results in \cite{DBLP:phd/hal/Lopez23a}
do not apply in our setting.

\begin{theorem}\label{embedding_wqo}
The embedding relation on any class of pointed tree-shaped models whose height is uniformly bounded by a fixed constant
is a well-quasi-order.
\end{theorem}
\begin{proof}
Let $L$ be the uniform bound on the height of the models under consideration, and let $\preceq$ be the embedding relation between pointed models.
Clearly, $\preceq$ is a quasi-order.
We prove by induction on $\ell\le L$ the following claim:
for every class $\T$ of pointed  tree-shaped  models of height at most $\ell$,
the quasi-order $(\T,\preceq)$ is a wqo.

\smallskip
\noindent
\emph{Base case ($\ell=0$)}. A tree-shaped model of height $0$ consists of a single root labelled by some set
$S \subseteq \prop$. Since the propositional signature $\prop$ is finite, there are finitely many non-isomorphic  models of this form. 
Hence, for any class of pointed tree-shaped models of height $0$, the quasi-order 
induced by $\preceq$ is a wqo.  Note  that the finiteness of 
$\prop$ is essential for this argument.

\smallskip
\noindent
\emph{Inductive step}. Let $\ell \geq 1$ and assume that the claim holds for $\ell-1$. 
Let $\T$ be an arbitrary class of pointed tree-shaped models of height at most $\ell$. 
We show that $(\T, \preceq)$ is a wqo.  

We begin by partitioning $\T$
into subclasses $\T^S$ indexed by $S \subseteq \prop$, where $\T^S$ consists of all pointed models in $\T$ whose root satisfies exactly the propositions in $S$.
Since $2^\prop$ is finite and a finite disjoint union of wqos is also a wqo,
it suffices to show that each $(\T^S, \preceq)$ is a wqo. 

Fix $S \subseteq \prop$. For each $(\M, w)$ in $\T^S$, let 
$\tau_{\M,w} = (\tau_1, \dots, \tau_n)$
be a finite sequence enumerating its immediate subtrees rooted at the children of $w$.
As the children are unordered, we fix an arbitrary but canonical ordering on models
to obtain a well-defined enumeration.
The embedding relation $\preceq$ on pointed models induces the subsequence order relation $\preceq^*$ on their finite sequences, defined  as in Higman's Lemma (Lemma \ref{higman}).
To account for the arbitrary ordering of children, 
we consider the permutation-closed variant $\preceq^*_{\pi}$
defined by $\tau \preceq^*_{\pi} \tau'$ if and only if $\tau \preceq^* f(\tau')$ for some permutation $f$.
We claim that for any $(\M,w),(\N,v)\in\T^S$, $(\M,w) \preceq (\N,v)$
if and only if $\tau_{\M,w} \preceq^*_\pi \tau_{\N,v}$.
Indeed, an embedding from $(\M,w)$ into $(\N,v)$ must map each child of $w$ injectively
to a distinct child of $v$, such that the corresponding subtree embeds recursively.
This yields $ \tau_{\M, w} \preceq^*_\pi \tau_{\N, v} $.
Conversely, if  $ \tau_{\M, w} \preceq^*_\pi \tau_{\N, v} $ then
each subtree in $\tau_{\M,w}$ embeds into a distinct 
subtree in $\tau_{\N,v}$, yielding an embedding from $(\M,w)$ into $(\N,v)$.

Thus, we show that $\bigl( \{\tau_{\M,w} \mid (\M,w)\in\T^S\}, \preceq^*_\pi \bigr)$  is a wqo.
To this end, let $\mathcal I$ be the set of all pointed tree-shaped models  of
height at most $\ell-1$.
By the inductive hypothesis,  $(\mathcal{I},\preceq)$ is a wqo. By  Higman's Lemma (Lemma \ref{higman}), the set of finite sequences $\mathcal{I}^*$ ordered by the subsequence relation $\preceq^*$ is also a wqo.
Since wqos are closed under taking subsets, and 
$(\tau_{\M,w} \mid (\M,w) \in \T^S)$ is a subset of $\mathcal I^*$,
it follows that $((\tau_{\M,w} \mid (\M,w) \in \T^S) , \preceq^*)$ is a wqo. Moreover, as $\preceq^*$ is a subrelation  of $\preceq^*_\pi$ and $\preceq^*_\pi$ is a quasi-order, it follows that $\preceq^*_\pi$ is also a wqo on this set. This completes the inductive step.
\end{proof}

As an immediate consequence, the preceding theorem extends to the other semantic relations considered in this paper. Since wqos are closed
under taking super-relations, and since the injective homomorphism relation and the homomorphism relation extend  the embedding relation, the same result holds for these relations.

\begin{corollary}\label{cor:finiteness}
The injective homomorphism relation, as well as the homomorphism relation,  on any class of  tree-shaped models whose
height is uniformly bounded by a constant is a  well-quasi-order.
\end{corollary}

As a result,  for $\preceq$ equal to the embedding, injective homomorphism, or homomorphism
relation, and for any class  $\T$ of pointed tree-shaped models of uniformly bounded height, the quasi-order 
$(\T,\preceq)$ is a wqo.
In particular, $\T$ admits only finitely many $\preceq$-minimal elements. This finiteness property underpins the minimal-model arguments used in the subsequent preservation theorems.

We conclude this section by showing that the bounded-height assumption 
is essential: if height is unbounded, the well-quasi-order property fails.

\begin{theorem}\label{unbounded_bad}
The embedding, injective homomorphism, and homomorphism relations on classes of pointed tree-shaped models of 
unbounded height are, in general, not well-quasi-orders.
\end{theorem}
\begin{proof}
We start with the case of injective homomorphisms, which also covers embeddings as a special case.
Consider the class $\T$ of pointed tree-shaped models that, for every $n\geq 1$, contains a pointed model $(\M_n, v_1)$ with worlds $v_1, \ldots v_n, v_{n+1}, v_{n+2},u,w$  such that $v_{i}Rv_{i+1}$ for $1 \leq i \leq n+1$ and
$v_1Ru$ and $v_{n+1}Rw$, as depicted in \Cref{antichain}. 
All worlds are unlabelled,  that is, the valuation function maps each propositional symbol to the empty 
set. 

In each $(\M_n, v_1)$ there are exactly two worlds at distance $n+1$ from the root, namely 
$v_{n+2}$ and $w$. 
By contrast, if $m > n$, then $(\M_m, v_1)$
has a single world at distance $n+1$, namely $v_{n+2}$.
Hence, for any $ n \neq m$, 
there is neither an injective homomorphism from $(\M_n, v_1)$ to $(\M_m, v_1)$, nor vice-versa.
Thus, $\T$ forms an infinite antichain with respect to injective homomorphisms
and hence also with respect to embeddings.

\begin{figure}[ht]
\centering
\begin{tikzpicture}
\tikzset{
>=latex,
node distance=0.8cm,
world/.style={
    draw,
    circle,
    fill,
    minimum size=1.5mm,
    inner sep=0pt,
    scale=1
  }
}

\begin{scope}[local bounding box=M1]
% ----- nodes -----
\node[world, label=right:$v_1$] (v1) at (0,0) {};
\node[world, below of=v1, label=right:$v_2$] (v2)  {};
\node[world, below of=v2, label=right:$v_3$] (v3)  {};

\node[world, below right of=v1, label=right:$u$] (u)  {};
\node[world, below right of=v2, label=right:$w$] (w)  {};
% ----- edges -----
\draw [->] (v1) --  (v2);
\draw [->] (v2) --  (v3);

\draw [->] (v1) --  (u);
\draw [->] (v2) --  (w);
\end{scope}
% ----- frame around the scope -----
\node[draw, rounded corners, fit=(M1), inner sep=6pt] (frame1) {};
% ----- label above -----
\node[above=2pt of frame1] {$\mathfrak{M}_1$};

\begin{scope}[xshift=2cm, local bounding box=M2]
% ----- nodes -----
\node[world, label=right:$v_1$] (v1) at (0,0) {};
\node[world, below of=v1, label=right:$v_2$] (v2)  {};
\node[world, below of=v2, label=right:$v_3$] (v3)  {};
\node[world, below of=v3, label=right:$v_4$] (v4)  {};

\node[world, below right of=v1, label=right:$u$] (u)  {};
\node[world, below right of=v3, label=right:$w$] (w)  {};
% ----- edges -----
\draw [->] (v1) --  (v2);
\draw [->] (v2) --  (v3);
\draw [->] (v3) --  (v4);

\draw [->] (v1) --  (u);
\draw [->] (v3) --  (w);
\end{scope}
% ----- frame around the scope -----
\node[draw, rounded corners, fit=(M2), inner sep=6pt] (frame2) {};
% ----- label above -----
\node[above=2pt of frame2] {$\mathfrak{M}_2$};

\begin{scope}[xshift=4cm, local bounding box=M3]
% ----- nodes -----
\node[world, label=right:$v_1$] (v1) at (0,0) {};
\node[world, below of=v1, label=right:$v_2$] (v2)  {};
\node[world, below of=v2, label=right:$v_3$] (v3)  {};
\node[world, below of=v3, label=right:$v_4$] (v4)  {};
\node[world, below of=v4, label=right:$v_5$] (v5)  {};

\node[world, below right of=v1, label=right:$u$] (u)  {};
\node[world, below right of=v4, label=right:$w$] (w)  {};
% ----- edges -----
\draw [->] (v1) --  (v2);
\draw [->] (v2) --  (v3);
\draw [->] (v3) --  (v4);
\draw [->] (v4) --  (v5);

\draw [->] (v1) --  (u);
\draw [->] (v4) --  (w);
\end{scope}
% ----- frame around the scope -----
\node[draw, rounded corners, fit=(M3), inner sep=6pt] (frame3) {};
% ----- label above -----
\node[above=2pt of frame3] {$\mathfrak{M}_3$};

% ----- dots -----
% \node[above right=6pt and 0.5cm of frame3] {$\cdot\ \cdot\ \cdot$};
\node[above right=-1.5cm and 0.5cm of frame3] {$\cdot\ \cdot\ \cdot$};

\end{tikzpicture}
\caption{An infinite antichain of models with respect to the injective homomorphism (and embedding) relation}\label{antichain}
\end{figure}
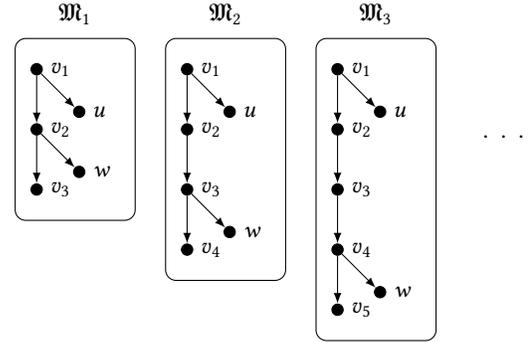

We now turn to the case of homomorphisms. Consider  the modification of $\T$ where, in each model
$(\M_n, v_1)$, the 
valuation function is defined as: $V(p_1) = \{v_i \mid i \text{ is even } \}$, 
$V(p_2) = \{v_i \mid i \text{ is odd } \}$, $V(p_3) = u$, and $V(p_4) = w$; all other propositions 
evaluate to $\emptyset$. 

Let $n<m$ and suppose, towards a contradiction, that $f$ is a homomorphism  from 
$(\M_n,v_1)$ to $(\M_m,v_1)$.
We prove by induction on $i\le n+2$ that $f(v_i)=v_i$. 
The base case $i=1$ is immediate.
For the inductive step, assume that $f(v_i)=v_i$. 
Then, $f(v_{i+1})$ must be a successor of $v_i$ in $\M_m$ satisfying the same propositions as $v_{i+1}$ in $\M_n$. 
In $\M_m$ 
the only such successor  is $v_{i+1}$ itself, hence $f(v_{i+1})=v_{i+1}$.
In particular, $f(v_{n+1})=v_{n+1}$. It follows that $f(w)$ is a successor of $v_{n+1}$ in $\M_m$  satisfying $p_4$. 
However, the  only world 
in $\M_m$ satisfying $p_4$  is  $w$, which is a successor of
$v_{m+1}$ rather than $v_{n+1}$. This
yields a contradiction.  
A symmetric argument shows that no homomorphism exists from  $(\M_m, v_1)$ to $(\M_n, v_1)$.
\end{proof}

As a final observation, we note that the proof of \Cref{unbounded_bad} relies on two distinct constructions. 
The counterexample for  embeddings and injective homomorphisms uses \emph{unlabelled} models, that is, models
in which all propositional symbols are evaluated as the empty set.
In contrast, the counterexample for homomorphisms relies on \emph{labelled} models.
As we show next, the use of  labelled models in the second construction is unavoidable. 

\begin{proposition}
For any pair of  unlabelled pointed tree-shaped  models, there exists a homomorphism from one model into the other.
\end{proposition}
\begin{proof}
Let  $(\M,w)$ and $(\N,v)$ be unlabelled pointed models of heights $m$ and $n$ respectively, with
$m \leq n$.
Since $\N$ has height $n$, it contains a path $(v_1,v_2, \dots, v_{n})$ of length $n$, with $v_1 =v$.
We define a function $f$ from the worlds of $\M$ to the worlds of $\N$ by mapping
each world $w'$ of $\M$ at distance $i$ from the root to $v_i$.
This mapping is well defined since $i \leq m \leq n$.
Because both models are unlabelled, the valuation condition for homomorphisms is trivially satisfied.
Moreover, by construction, $f$ preserves the edge relation.
Hence, $f$ is a homomorphism from $(\M,w)$ to $(\N,v)$.
\end{proof}

%%%%%%%%%%%%%%%%%%%%%
\section{Preservation Theorems for Unravelling-Invariant Classes and Modal Logics}\label{sec:pres}
%%%%%%%%%%%%%%%%%%%%%

We now use the results from the previous section---most notably,  \Cref{embedding_wqo}---to establish preservation theorems for classes of models that are invariant under unravelling.
These results will subsequently be applied 
to obtain preservation theorems for modal logics and, in turn, for graph neural networks.

This section is organised into three main parts,  which
relate unravelling-invariant classes of models preserved under embeddings, injective homomorphisms, and homomorphisms, with definability in  \EGML{},  \EPGML{}, and  \EPML{}, respectively.

All three preservation theorems
follow the same general proof strategy. 
Let $\C$ be a class of models invariant under $L$-unravelling and preserved under a relation $\preceq$ (namely, embedding, injective homomorphism, or homomorphism).
We first construct the class $\T$ of  $L$-unravellings of models in $\C$. 
We then let $\Tmin$ be the set of $\preceq$-minimal models in $\T$.
For each model $(\mathfrak T,v) \in \Tmin$ we 
construct a characteristic formula  $\varphi_{\mathfrak T,v}$ such that, for every pointed model $(\M,w)$,  $(\M,w) \models \varphi_{\mathfrak T,v}$ 
if and only if $(\mathfrak T, v) \preceq \unr^L(\M, w)$.
Finally, we define a formula $\varphi_\C$ as a disjunction over all $\varphi_{\mathfrak T,v}$ with $(\mathfrak T,v) \in \Tmin$.
By construction, the formula $\varphi_\C$ defines the class $\C$.

Finiteness of $\Tmin$ 
is crucial for this strategy.
If $\Tmin$ were infinite, the resulting disjunction  would be infinite.
Importantly, the finiteness of $\Tmin$ is guaranteed by \Cref{embedding_wqo}.
The modal language to which the characteristic formulae 
$\varphi_{\mathfrak T,v}$ belong
depends on the chosen relation  $\preceq$:  for $\preceq$ being the
embedding, injective homomorphism, and homomorphism,  the formulae
$\varphi_{\mathfrak T,v}$ belong, respectively, to the fragments \EGML{}, \EPGML{}, and \EPML{}.

%%%%%%%%%%%%%%%%%%%%%
\subsection{Preservation Under Embeddings}
%%%%%%%%%%%%%%%%%%%%%

Our aim is to show that any class of models preserved under embeddings and invariant under $L$-unravelling is definable by a $\EGML$-formula. Moreover, we show
that such a formula can be chosen to have depth at most $L$.
To this end, we introduce characteristic formulae
in $\EGML$. These formulae are obtained  as a modification of the standard characteristic formulae for graded modal logic from the literature \cite{DBLP:journals/corr/abs-1910-00039,boundedGNNs}.
The key difference is that our characteristic formulae in $\EGML$ do not contain negated modal operators.

\begin{definition}\label{def:characteristic-existential}
Let $(\M,w)$ be a pointed model with $\M = (W,R,V)$.
For each $\ell \in \mathbb{N}$, we define
a \emph{characteristic $\EGML$ formula}, $\varphi_{\M,w}^\ell$  of depth $\ell$ inductively as follows.
For $\ell = 0$, let 
$$
\varphi_{\M,w}^0  := \bigwedge \{ p  \mid (\M,w) \models p \} 
\land
\bigwedge \{ \neg p  \mid (\M,w) \not\models p \}.
$$
For $\ell + 1$, partition the $R$-successors of $w$ into  equivalence classes $C_1, \dots, C_m$ according to the relation 
$\equiv^{\GML}_\ell$.
For each $C_j$, choose an arbitrary representative $w_j \in C_j$
and define 
$$
\varphi_{\M,w}^{\ell+1}
\ :=\
\varphi_{\M,w}^0 
\ \wedge\
\bigwedge_{j=1}^m \Diamond^{\geq |C_j|}\,\varphi_{\M,w_j}^\ell .$$
\end{definition}

For full \GML{}, a characteristic formula of depth $\ell$ for a pointed model
$(\M,w)$ holds in $(\N,v)$ if and only if there exists a graded $\ell$-bisimulation between  these two pointed models. Over
finite models, this is equivalent to $(\M,w) \equiv^\GML_\ell (\N,v)$ \cite{DBLP:journals/corr/abs-1910-00039}. 
The characteristic formulae defined above for \EGML{} are strictly weaker, and
these equivalences no longer hold.
Instead, the \EGML{} characteristic formulae ensure a one-way structural property: the $\ell$-unravelling of the  model used to construct the characteristic formula embeds into the $\ell$-unravelling of any  model satisfying the formula. 
This property is captured by the following lemma, which plays a central role in the 
preservation theorem for embeddings.

\begin{lemma}\label{lem:preserve}
Let $(\M, w)$  and $(\N, v)$ be pointed models and let  $\ell \in \mathbb{N}$.
We have $(\N,v) \models \varphi_{\M,w}^\ell$ 
if and only if $\unr^\ell(\M, w)$ embeds into $\unr^\ell(\N, v)$.
\end{lemma}
\begin{proof}
We prove the equivalence by induction on $\ell$. Let $\preceq$ be the
embedding relation between pointed models.

\smallskip
\noindent
\emph{Base case ($\ell=0$)}. Formula $\varphi_{\M,w}^0$ is a conjunction of all literals  (propositions and their negations) satisfied at $(\M,w)$. 
Thus, $(\N, v) \models \varphi_{\M,w}^0$ iff $(\M,w) \equiv^\GML_0 (\N, v)$.
Since $\unr^0(\M,w)$ and $\unr^0(\N,v)$ consist single worlds, it holds iff
 $\unr^0(\M, w) \preceq \unr^0(\N, v)$. 

\smallskip
\noindent
\emph{Inductive step.} Assume that the claim holds for $\ell \geq 0$. 
We prove that it also holds for $\ell+1$. 
Recall that
$
\varphi_{\M,w}^{\ell+1}
\ :=\
\varphi_{\M,w}^0 
\ \wedge\
\bigwedge_{j=1}^m \Diamond^{\geq |C_j|}\,\varphi_{\M,w_j}^\ell$, where 
$C_1, \dots, C_m$ are the equivalence classes of
$R$-successors of $w$
under
$\equiv^{\GML}_\ell$ , and  $w_j \in C_j$ is an
arbitrary representative. For each $j$, we also fix an
enumeration
$C_j = \{w_{j,1}, \ldots, w_{j,|C_j|} \}$.

\smallskip\noindent
\emph{($\Rightarrow$)}
Assume that $(\N,v) \models \varphi_{\M,w}^{\ell+1}$. 
Then, $(\N,v) \models \varphi_{\M,w}^{0}$, and hence $(\M,w) \equiv^\GML_0 (\N, v)$. 
Moreover, for each $j$, there exist $|C_j|$ distinct successors 
$v_{j,1}, \ldots, v_{j,|C_j|}$ of $v$ such that
$(\N, v_{j,k}) \models \varphi_{\M,w_j}^\ell$.

By the inductive hypothesis, for each such  $j$ and $k$, there exists an embedding $f_{j,k}$ from $\unr^\ell(\M, w_j)$ into $\unr^\ell(\N, v_{j,k})$. 
Since $w_{j,k} \equiv^{\GML}_{\ell} w_j$, the unravellings
$\unr^{\ell}(\M,w_{j,k})$ and $\unr^{\ell}(\M,w_j)$ are isomorphic.
Composing this isomorphism with $f_{j,k}$ yields an embedding of
$\unr^{\ell}(\M,w_{j,k})$ into $\unr^{\ell}(\N,v_{j,k})$.

We now construct a global embedding $f$ 
from $\unr^{\ell+1}(\M, w)$  into  $\unr^{\ell+1}(\N, v)$  as follows. 
Function $f$ maps the root of $\unr^{\ell+1}(\M, w)$ to the root of $\unr^{\ell+1}(\N, v)$. For each $j$ and $k$, $f$ coincides with the embedding constructed above on the 
subtree rooted at the node corresponding to $w_{j,k}$.
Since the successors $v_{j,k}$ are pairwise distinct, the images of these embeddings are
disjoint. It follows that $f$ is   an embedding of $\unr^{\ell+1}(\M,w)$ into $\unr^{\ell+1}(\N,v)$.

\smallskip
\noindent
\emph{($\Leftarrow$)}
Conversely, assume that $f$ is an embedding of  $\unr^{\ell+1}(\M, w)$ into $\unr^{\ell+1}(\N, v)$. 
Then, $f$ maps the root of $\unr^{\ell+1}(\M,w)$ to the root of $\unr^{\ell+1}(\N,v)$, and therefore, $(\M,w)\equiv^{\GML}_0(\N,v)$.
Thus,  $(\N, v) \models \varphi_{\M,w}^0$.
Fix $j \in \{1, \ldots, m\}$. 
By  the definition of unravelling,
the root of $\unr^{\ell+1}(\M,w)$ has exactly $|C_j|$
distinct nodes at depth $1$ corresponding to the successors
$w_{j,1},\dots,w_{j,|C_j|}$ of $w$.
Since $f$ is injective, it maps these nodes 
to $|C_j|$ distinct nodes at depth $1$ in $\unr^{\ell+1}(\N,v)$.
Each such  node  corresponds to a 
successor
$v_{j,k}$ of $v$ in $\N$, and the subtree below this node is precisely
$\unr^\ell(\N,v_{j,k})$.
Restricting $f$ to the corresponding subtrees yields embeddings, so
$\unr^\ell(\M,w_{j,k}) \preceq \unr^\ell(\N,v_{j,k})$
for $k=1,\dots,|C_j|$.
By the inductive hypothesis, this implies $(\N, v_{j,k}) \models \varphi_{\M,w_{j,k}}^\ell$.
Since each $w_{j,k}  \in C_j$, we obtain that $\varphi_{\M,w_{j,k}}^\ell$ is identical to $\varphi_{\M,w_{j}}^\ell$.
Moreover, the worlds $v_{j,k}$ are pairwise distinct successors
of $v$. Hence,
$(\N,v) \models 
\Diamond^{\geq |C_j|}\,\varphi_{\M,w_j}^\ell$, as required.
\end{proof}

Equipped with \Cref{lem:preserve} and \Cref{embedding_wqo}, we are  now ready to prove the preservation theorem for classes of models preserved under embeddings and invariant under $L$-unravelling.
As we show below, these are exactly the classes definable by \EGML{}-formulae of modal depth  at most $L$.

\begin{theorem}\label{th:preserve}
The following are equivalent for any class  $\C$ of pointed models and any $L\in \mathbb{N}$:
\begin{itemize}
\item $\C$ is invariant under $L$-unravelling and preserved under embeddings,
\item $\C$ is definable by an $\EGML$ formula of depth at most $L$.
\end{itemize}
\end{theorem}

\begin{proof}

\smallskip
\noindent
\emph{($\Rightarrow$)}
Let $\T = \{ \unr^L(\M, w) \mid (\M,w) \in \C\}$ and let $\preceq$ be the embedding relation. 
Since $\C$ is invariant under
$L$-unravelling, we have that $\T \subseteq \C$.
By  \Cref{embedding_wqo}, $(\T, \preceq)$ is a wqo and therefore admits
only finitely many $\preceq$-minimal elements up to isomorphism. Let \Tmin contain exactly
one representative for each such minimal isomorphism type. 
For each $(\mathfrak T,v) \in \Tmin$, we construct the $\EGML$ characteristic formula $\varphi_{\mathfrak T,v}^L$ and define 
$\varphi_{\C} = \bigvee_{(\mathfrak T,v) \in \Tmin} \varphi_{\mathfrak T,v}^L$. 
Since the  disjunction is finite and each $\varphi_{\mathfrak T,v}^L$ is an \EGML{}-formula, $\varphi_{\mathcal{C}}$ is a (finite)  \EGML{}-formula. 
We show that $\varphi_{\C}$ defines $\C$.

Let $(\M, w) \in \C$.  Then  $\unr^L(\M, w) \in \T$. Since 
$(\T,\preceq)$ is a wqo, there exists a $\preceq$-minimal element $(\mathfrak T,v) \in \Tmin$ such that $(\mathfrak T,v) \preceq \unr^L(\mathfrak M, w)$.
Because $(\mathfrak T,v)$ is tree-shaped of height at most $L$, we have $\unr^L(\mathfrak T,v) = (\mathfrak T,v) $.
Hence, $\unr^L(\mathfrak T,v) \preceq \unr^L(\mathfrak M, w)$ and 
by \Cref{lem:preserve} we obtain  $(\M, w) \models \varphi_{\mathfrak T,v}^L$.
It follows that $(\M,w) \models \varphi_{\C}$.

Conversely, suppose that $(\M,w) \models  \varphi_{\C}$. Then, $(\M,w) \models  \varphi_{\mathfrak T,v}^L$
for some $(\mathfrak T,v) \in \Tmin$. 
By \Cref{lem:preserve}, this implies  $\unr^L(\mathfrak T,v) \preceq \unr^L(\M, w)$.
Since $(\mathfrak{T},v)$ is tree-shaped of height at most $L$, we again have
$\unr^L(\mathfrak{T},v)=(\mathfrak{T},v)$, and hence
$(\mathfrak{T},v)\preceq \unr^L(\M,w)$.
By construction, $\Tmin \subseteq \T \subseteq \C$, so
$(\mathfrak T,v) \in \C$. Since $\C$ is preserved under embeddings, we have $\unr^L(\M, w) \in \C$.
Finally, invariance of $\C$ under $L$-unravelling yields $(\M, w) \in \C$.

\smallskip
\noindent
\emph{($\Leftarrow$)} Suppose that $\C$ is definable by an $\EGML$-formula $\varphi$ of
depth at most $L$. Then $\C$ is invariant under $L$-unravellings: for every pointed model,
$(\M,w)\models\varphi$ if and only if $\unr^L(\M,w)\models\varphi$
since the truth of $\varphi$ depends only on the unravelling up to depth $L$.
We argue that $\C$ is also preserved under embeddings.
Let $(\M,w)\models\varphi$ and suppose that there exists an embedding
witnessing $(\M,w)\preceq(\N,v)$. Since embeddings preserve propositional labels and map successors injectively, a routine induction on the structure of $\varphi$ shows that $(\N,v) \models \varphi$. In particular, for each subformula of the form $\Diamond^{\ge k}\psi$ occurring in $\varphi$, 
the existence of $k$ successors satisfying $\psi$ in $(\M,w)$ implies the existence of
$k$ successors satisfying $\psi$ in $(\N,v)$. Thus, $\C$ is preserved under embeddings.
\end{proof}

As a consequence of \Cref{th:preserve}, we obtain the following
equirank embedding preservation theorem for $\GML$.

\begin{theorem}\label{embeddingGML}
The following are equivalent for any \GML{} formula $\varphi$ of depth at most $L$:
\begin{itemize}
\item $\varphi$ is preserved under embeddings,
\item $\varphi$ is equivalent to a $\EGML$ formula of depth at most $L$.
\end{itemize}
\end{theorem}
\begin{proof}
\smallskip
\noindent
\emph{($\Rightarrow$)} Let $\varphi$ be a $\GML$ formula of depth at most $L$ that is preserved under
embeddings. Every $\GML$ formula of depth at most $L$ is invariant under $L$-unravelling.
Hence, by Theorem~\ref{th:preserve}, $\varphi$ is equivalent to an $\EGML$ formula
of depth at most $L$.

\smallskip
\noindent
\emph{($\Leftarrow$)} Conversely, suppose that $\varphi$ is equivalent to a $\EGML$ formula of depth at most $L$. By the $(\Leftarrow)$ direction of \Cref{th:preserve}, every $\EGML$ formula
is preserved under embeddings. Since preservation under embeddings is invariant under logical equivalence, $\varphi$ is preserved under embeddings.
\end{proof}

%%%%%%%%%%%%%%%%%%%%%
\subsection{Preservation Under Injective Homomorphisms}
%%%%%%%%%%%%%%%%%%%%%

We next characterise unravelling-invariant  classes  that are preserved under injective homomorphisms.
Recall that every embedding is an injective homomorphism, but not conversely. 
This suggests that the logical language defining classes preserved under injective homomorphisms should be 
strictly weaker than the one required for preservation under embeddings.
By \Cref{th:preserve}, the latter language is \EGML.
As we show next, the appropriate language in the case of injective homomorphisms is \EPGML{}, the existential-positive fragment of \GML{}. 

We first define characteristic $\EPGML$-formulae,  obtained from the characteristic  \EGML{} formulae by deleting all negative literals.

\begin{definition}
Let $(\M,w)$ be a pointed model with $\M = (W,R,V)$. The \emph{characteristic $\EPGML$-formula} $\psi_{\M,w}^{\ell}$ of depth $\ell$  is defined as the formula obtained from the
characteristic $\EGML$-formula $\varphi_{\M,w}^{\ell}$
(see \Cref{def:characteristic-existential}) by removing all negative literals $\neg p$.
\end{definition}

Note that  $\psi_{\M,w}^{\ell}$ is defined  in the same way as
$\varphi_{\M,w}^{\ell}$, except that the conjunction
$\bigwedge \{ \neg p \mid (\M,w)\not\models p \}$
is omitted.
As a result, $\psi_{\M,w}^{\ell}$ contains no negations and is therefore an \EPGML{}-formula.

While $\varphi_{\M,w}^{\ell}$ guarantees the existence of an embedding between $\ell$-unravellings of models, the formula $\psi_{\M,w}^{\ell}$ guarantees the existence of an injective homomorphism between these unravellings. This is captured by the following lemma.

\begin{lemma}\label{lem:preserve-injective}
Let $(\M, w)$  and $(\N, v)$ be pointed models and let  $\ell \in \mathbb{N}$. Then,
$(\N,v) \models \psi_{\M,w}^\ell$ 
if and only if there exists an injective homomorphism from $\unr^\ell(\M, w)$ to $\unr^\ell(\N, v)$.
\end{lemma}
\begin{proof}
The proof is by induction on $\ell$ and closely parallels the proof of Lemma~\ref{lem:preserve};
we highlight only the points where injective homomorphisms replace embeddings.
Throughout this proof, let $\preceq$
be the injective homomorphism relation.

\smallskip
\noindent
\emph{Base case ($\ell=0$).} The formula
$\psi^0_{\M,w}$ is the conjunction of all propositional symbols satisfied at $w$. Thus,
$(\N,v)\models\psi^0_{\M,w}$ if and only if 
$(\N,v)$ satisfies at least the same propositions as $(\M,w)$.
Since $\unr^0(\M,w)$ and $\unr^0(\N,v)$ each consist of a single world, this
holds if and only if 
$\unr^0(\M,w) \preceq \unr^0(\N,v)$.

\smallskip
\noindent
\emph{Inductive step.} 
Assume that the claim holds for some $\ell \geq 0$, and consider $\ell+1$.
We use the same notation as in the proof of \Cref{lem:preserve}. 

\smallskip
\noindent
\emph{($\Rightarrow$)}
Assume that $(\N,v) \models \psi^{\ell+1}_{\M,w}$. Then $(\N,v)\models\psi^0_{\M,w}$, and for each equivalence
class $C_j$, there exist $|C_j|$  distinct successors $v_{j,1},\dots,v_{j,|C_j|}$ of $v$ such that
$(\N,v_{j,k})\models\psi^\ell_{\M,w_j}$. By the inductive hypothesis, for each $j,k$,
there exists an injective homomorphism and hence $\unr^\ell(\M,w_j) \preceq \unr^\ell(\N,v_{j,k})$. As in the proof of 
\Cref{lem:preserve}, these homomorphisms can be combined together into a single injective homomorphism from $\unr^{\ell+1}(\M,w)$ to $\unr^{\ell+1}(\N,v)$  by
mapping the roots and using the homomorphisms above on the corresponding subtrees. 
Since the targets lie in pairwise disjoint subtrees rooted at distinct depth-$1$ nodes,
the resulting homomorphism is injective.
Thus, $\unr^{\ell+1}(\M,w) \preceq \unr^{\ell+1}(\N,v)$.

\smallskip
\noindent
\emph{($\Leftarrow$)}
Conversely, suppose that there exists an injective homomorphism
$f$ from $\unr^{\ell+1}(\M,w)$ to $\unr^{\ell+1}(\N,v)$. Then $f$ maps the root to the root, which implies $(\N,v)\models\psi^0_{\M,w}$.
Fix $j \in \{1,\ldots,m\}$. Injectivity of $f$ forces the $|C_j|$ distinct
depth-$1$ nodes below the root of $\unr^{\ell+1}(\M,w)$ to be mapped 
to $\mid C_j \mid$ distinct nodes of depth 1 of  $\unr^{\ell+1}(\N,v)$.
Each such node corresponds to a successor $v_{j,k}$ of $v$ in $\N$.
Restricting $f$ to the corresponding subtrees yields injective homomorphisms from $\unr^\ell(\M,w_{j,k})$ to $\unr^\ell(\N,v_{j,k})$. 
By the induction hypothesis we obtain $(\N,v_{j,k}) \models \psi^\ell_{\M,w_j}$ 
for all $k$. Since the worlds $v_{j,k}$ are pairwise distinct successors of $v$, we conclude that $(\N,v) \models \Diamond^{\ge |C_j|}\psi^\ell_{\M,w_j}$.
As this holds for each $j$, $(\N,v) \models \psi^{\ell+1}_{\M,w}$.
\end{proof}

Using the results above on characteristic $\EGML$-formulae, we can now establish the following preservation theorem for unravelling-invariant classes preserved under injective homomorphisms.

\begin{theorem}\label{th:exists-positive-unravelling}
The following are equivalent for any class  $\C$ of pointed models and any $L\in \mathbb{N}$:
\begin{itemize}
\item $\C$ is preserved under injective homomorphisms and invariant under $L$-unravelling,
\item $\mathcal{C}$ is definable by an $\EPGML$-formula of depth at most $L$.
\end{itemize}

\end{theorem}
\begin{proof}

The proof follows the same strategy as that of Theorem \ref{th:preserve}, but now exploits \Cref{lem:preserve-injective} instead of \Cref{lem:preserve}.

\smallskip
\noindent
\emph{($\Rightarrow$)}
Let $\T=\{\,\unr^L(\M,w)\mid(\M,w)\in\C\,\}$ and
let $\preceq$ be the injective homomorphism relation.
Since $\C$ is invariant under $L$-unravellings, we have $\T \subseteq \C$. By  \Cref{cor:finiteness}, $(\T, \preceq)$ is a wqo and hence admits only finitely many $\preceq$-minimal elements up to isomorphism. Let \Tmin contain 
one representative for each such minimal isomorphism type. 
For each $(\mathfrak T,v) \in \Tmin$, we construct the $\EPGML$ characteristic formula $\psi_{\mathfrak T,v}^L$ and define 
$\psi_{\C} = \bigvee_{(\mathfrak T,v) \in \Tmin} \psi_{\mathfrak T,v}^L$. 
The argument showing that $\psi_{\C}$ defines $\C$ is identical to the corresponding one in the proof of  \Cref{th:preserve}, with \Cref{lem:preserve-injective} used in place of    \Cref{lem:preserve}.

\smallskip
\noindent
\emph{($\Leftarrow$)} Suppose that $\C$ is definable by an $\EPGML$-formula $\varphi$ of
depth at most $L$. Then $\C$ is invariant under $L$-unravellings: for every
pointed model  $(\M,w)$,  
$(\M,w)\models\varphi$ if and only if $\unr^L(\M,w)\models\varphi$. 
We now show that $\C$ is preserved under injective homomorphisms.
Let $(\M,w)\models\varphi$ and let $f$ be an
injective homomorphism
 witnessing $(\M,w)\preceq(\N,v)$. Since injective homomorphisms  preserve propositional labels monotonically and map successors injectively, a routine induction on the structure of $\varphi$ shows that $(\N,v) \models \varphi$, exactly as in the proof of Theorem \ref{th:preserve}.
\end{proof}

We now use \Cref{th:exists-positive-unravelling} to show the equirank preservation theorem for \GML{} with respect to injective homomorphisms.

\begin{theorem}\label{injectiveGML}
The following are equivalent for any \GML{}-formula $\varphi$ of depth at most $L$:
\begin{itemize}
\item $\varphi$ is preserved under injective homomorphisms,
\item $\varphi$ is equivalent to a $\EPGML$-formula of depth at most $L$.
\end{itemize}
\end{theorem}
\begin{proof}
\smallskip
\noindent
\emph{($\Rightarrow$)} Let $\varphi$ be a $\GML$-formula of depth at most $L$ that is preserved under
injective homomorphisms. Every $\GML$-formula of depth at most $L$ is invariant under $L$-unravelling.
Hence, by Theorem~\ref{th:exists-positive-unravelling}, $\varphi$ is equivalent to an $\EPGML$-formula
of depth $ \leq L$.

\smallskip
\noindent
\emph{($\Leftarrow$)} Suppose that $\varphi$ is equivalent to an $\EPGML$-formula of depth at most $L$. By the $(\Leftarrow)$ direction of \Cref{th:exists-positive-unravelling}, every $\EPGML$-formula
is preserved under injective homomorphisms. Since preservation under injective homomorphisms is invariant under logical equivalence, $\varphi$ is preserved under injective homomorphisms.
\end{proof}

It is worth noting that logics corresponding
to both $\GML$ and $\EPGML$
have been studied extensively in the setting of description logics in knowledge representation and reasoning. 
In particular, the description logic corresponding to $\GML$ is $\ALCQ$ \cite{DBLP:conf/dlog/2003handbook}, while the 
logic corresponding to  $\EPGML$ is $\ELUQ$ \cite{morris2025sound}.
As a consequence, \Cref{injectiveGML} shows that  $\ELUQ$ is exactly the fragment of $\ALCQ$ that is preserved under injective homomorphisms.

We now turn to preservation under injective homomorphisms for $\ML$.
To establish the corresponding preservation theorem, we introduce
additional machinery, namely a pruning operation on tree-shaped models and some
of its basic properties.

\begin{definition}\label{def:pruning}
Let $(\M,w)$ be a pointed  tree-shaped  model of height $L$, and let $<$ be a fixed  linear order on the worlds of $\M$. 
We define a sequence of models $\M^{k}$ inductively as follows:
\begin{itemize}
\item $\M^{0} := \M$,
\item $\M^{k+1}$ is obtained
from $\M^{k}$ by removing any subtree $(\mathfrak{T},v)$ of height $k$ for which there is a  sibling subtree $(\mathfrak{T}',v')$ of height $k$ satisfying $(\mathfrak T', v') \equiv_k^{\ML} (\mathfrak T, v)$, and $v' < v$.
\end{itemize}
The \emph{pruning} of $(\M,w)$ is defined as $(\M^L,w)$. 
\end{definition}

An example of a model and its pruning is shown in \Cref{pruning}.
By construction, the pruning of a tree-shaped model is tree-shaped since it is obtained by deleting subtrees.
Moreover, the choice of the linear order $<$ in \Cref{def:pruning} does not affect the resulting pruning up to isomorphism: at each step, exactly one representative from each $\ML_k$-equivalence class of sibling subtrees is retained.

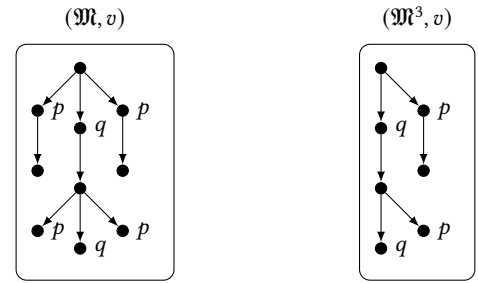
\begin{figure}[ht]
\centering
\begin{tikzpicture}
\tikzset{
>=latex,
node distance=0.8cm,
world/.style={
    draw,
    circle,
    fill,
    minimum size=1.5mm,
    inner sep=0pt,
    scale=1
  }
}

\begin{scope}[local bounding box=M1]
% ----- nodes -----
\node[world, label=right:] (v1) at (0,0) {};
\node[world, below left of=v1, label=right:$p$] (v2)  {};
\node[world, below of=v1, label=right:$q$] (v3)  {};
\node[world, below right of=v1, label=right:$p$] (v4)  {};

\node[world, below of=v3, label=right:] (u)  {};

\node[world, below  of=v2, label=right:] (v2')  {};
\node[world, below of=v4, label=right:] (v4')  {};

\node[world, below left of=u, label=right:$p$] (w2)  {};
\node[world, below of=u, label=right:$q$] (w3)  {};
\node[world, below right of=u, label=right:$p$] (w4)  {};

% ----- edges -----
\draw[->] (v1) -- (v2);
\draw[->] (v1) -- (v4);
\draw[->] (v1) -- (v3);

\draw[->] (v3) -- (u);

\draw[->] (v2) -- (v2');
\draw[->] (v4) -- (v4');

\draw[->] (u) -- (w2);
\draw[->] (u) -- (w4);
\draw[->] (u) -- (w3);
\end{scope}
% ----- frame around the scope -----
\node[draw, rounded corners, fit=(M1), inner sep=6pt] (frame1) {};
% ----- label above -----
\node[above=2pt of frame1] {$(\mathfrak{M},v)$};

\begin{scope}[xshift=4cm, local bounding box=M2]
% ----- nodes -----
\node[world, label=right:] (v1) at (0,0) {};
%\node[world, below left of=v1, label=right:$p$] (v2)  {};
\node[world, below of=v1, label=right:$q$] (v3)  {};
\node[world, below right of=v1, label=right:$p$] (v4)  {};

\node[world, below of=v3, label=right:] (u)  {};

\node[world, below of=v4, label=right:] (v4')  {};

%\node[world, below left of=u, label=right:$p$] (w2)  {};
\node[world, below of=u, label=right:$q$] (w3)  {};
\node[world, below right of=u, label=right:$p$] (w4)  {};
% ----- edges -----
%\draw[->] (v1) -- (v2);
\draw[->] (v1) -- (v4);
\draw[->] (v1) -- (v3);

\draw[->] (v3) -- (u);

\draw[->] (v4) -- (v4');

%\draw[->] (u) -- (w2);
\draw[->] (u) -- (w4);
\draw[->] (u) -- (w3);
\end{scope}
% ----- frame around the scope -----
\node[draw, rounded corners, fit=(M2), inner sep=6pt] (frame2) {};
% ----- label above -----
\node[above=2pt of frame2] {$(\mathfrak{M}^3,v)$};
% \node[above=15pt of frame2] {prunning};

\end{tikzpicture}
\caption{A tree-shaped model and its pruning} 
\label{pruning}
\end{figure}

Next, we establish two key properties of the pruning operation.
First, we show that a tree-shaped model of height $L$ and its pruning satisfy the same $\ML$ of depth at most $L$, that is, they are $L$-bisimilar. Second,
we show that the pruning admits an injective homomorphism into the
original model.

\begin{lemma}\label{lem:bisimilar-trees}
Let $(\M,w)$ be a pointed tree-shaped model of height $L$ and let $(\M^L,w)$ be 
its pruning.  The following hold:
\begin{itemize}
\item $(\M,w) \sim_{L} (\M^L,w)$,
\item there is an injective homomorphism from $(\M^L,w)$ to $(\M,w)$.
\end{itemize}
\end{lemma}

\begin{proof}
Let $\M = (W,R,V)$, and let $(\M^{L},w)$ be the pruning of $(\M,w)$. By construction, $W^L \subseteq W$, and each  $v \in W$ either survives the pruning (that is, $v \in W^L$) or is removed at some stage $k \leq L$ because a $<$-smaller sibling with the same $\equiv_k^{\ML}$-type was retained.

We define a relation $Z \subseteq W \times W^L$ as follows.
For each  $v \in W^L$, let $v Z v$.
For each $v \not\in W^L$, let $vZv'$, where $v'$ is the unique $<$-minimal sibling that caused $v$ to be removed during pruning.  By construction, $Z$ is a total function from $W$ to $W^L$.

We show that $Z$ is an $L$-bisimulation between $(\M,w)$ and $(\M^L,w)$. 

First, note that  
$R^L$ and $V^L$ are the restrictions of $R$ and $V$ to  $W^L$.
Moreover, if $vZv'$, then $(\M,v) \equiv^{\ML}_k (\M^L,v')$ for some $k \leq L$
by construction of the pruning. In particular, $V(v) = V^L(v')$.

\smallskip
\noindent
\emph{Forth condition of bisimulations.}
Assume $v Z v'$ and  $v R u$. By definition of $Z$, there exists $u' \in W^L$ such that $u Z u'$. 
We show that $v' R^L u'$.  

If $v \in W^L$, then $v' = v$. During pruning, the
successors of $v$ were partitioned into $\ML_k$-equivalence
classes, and one representative from
each class was retained. 
Since $u$ belongs to one such class, its retained representative
is precisely $u'$, and hence
$v'R^L u'$. 

If $v\notin W^{L}$, then 
$v'$ is a retained sibling of $v$ such that $(\M,v) \equiv^{\ML}_{k} (\M,v')$.
This equivalence implies that 
every $\ML_{k-1}$-type realised among the successors
of $v$ is also realised among the successors of $v'$.
Pruning ensures that the successor $u$ has a retained
representative $u'$ among the successors of $v'$, and by definition of $Z$ we have
$u Z u'$. Thus, also in this case, $v' R^{L} u'$.

\smallskip
\noindent
\emph{Back condition of bisimulations.}
Assume  $v Z v'$ and let $v' R^L u'$. By construction of the pruning, $u'$ is a retained representative of some subtree rooted at a successor $u$ of $v$ in $\M$.
By definition of $Z$, we have $u Z u'$, as required.

We conclude that $Z$ is an $L$-bisimulation between $(\M,w)$ and $(\M^{L},w)$, and therefore $(\M,w) \sim_{L} (\M^{L},w)$.

Finally, we show that the identity map $f: W^L \mapsto W$ is
an injective homomorphism from $(\M^L,w)$ to
$(\M,w)$.  Since $W^L \subseteq W$, this map is injective.
Moreover, for all $v',u' \in W^L$,
we have $v' R^L u'$ iff $v'R u'$ and $V^L(v') = V(v')$. 
Hence $f$ preserves both the accessibility relation and the valuation and is thus a 
homomorphism (indeed, an embedding) from $(\M^{L},w)$ into $(\M,w)$.
\end{proof}

Pruning retains exactly one representative from each $\ML_k$-equivalence class among sibling subtrees. As a consequence, no graded modality with counting greater than one can be satisfied in the pruned model.
This observation yields the following collapse  of characteristic formulae to basic modal logic.

\begin{lemma}\label{lem:ML-characteristic}
Let $(\M,w)$ be a pointed tree-shaped model of height $L$,
let $(\M^L,w)$ be its pruning, and let $\ell \leq L$.
Then, in both the  characteristic \EGML{}-formula $\varphi^{\ell}_{\M^L,w}$  and the characteristic \EPGML{}-formula $\psi^{\ell}_{\M^L,w}$, all operators $\Diamond^{ \geq k}$ occur only with $k=1$. In particular, $\varphi^{\ell}_{\M^L,w}$  is an $\EML$-formula and $\psi^{\ell}_{\M^L,w}$ is an \EPML{}-formula.
\end{lemma}

\begin{proof}
The proof is by induction on $\ell$. 

\smallskip
\noindent
\emph{Base case ($\ell=0$).}
The formulae
$\varphi^0_{\M^L, w}$ and $\psi^0_{\M^L, w}$ are propositional. Hence,
 $\varphi^0_{\M^{L},w}$ is an $\EML$-formula and $\psi^0_{\M^{L},w}$ is an $\EPML$-formula.

\smallskip
\noindent
\emph{Inductive step.} Let $\ell \geq 1$ and assume that the claim holds for 
$\ell-1$. Partition the children of $w$ in $\M^L$ into $\equiv^{\ML}_{\ell-1}$-equivalence classes $C_1, \dots, C_m$. 
By construction of the pruning, at most one 
representative from each $\equiv^{\ML}_{\ell-1}$-equivalent class is retained among sibling subtrees. Thus, each class $C_j$ is a singleton.

Since $\ML$ is a fragment of $\GML$, the relation $\equiv^{\GML}_{\ell-1}$ refines $\equiv^{\ML}_{\ell-1}$. 
It follows that each  $\equiv^{\GML}_{\ell-1}$-equivalence class among the 
children of $w$ in $\M^L$ is also a singleton. Therefore, the characteristic
$\EGML$-formula
$
\varphi_{\M^L,w}^{\ell}
\ :=\
\varphi_{\M^L,w}^0 
\ \wedge\
\bigwedge_{j=1}^m \Diamond^{\geq 1}\,\varphi_{\M^L,w_j}^{\ell-1}$ contains only
graded modal operators with threshold $1$.
By the inductive hypothesis, $\varphi_{\M^L,w_j}^{\ell-1}$ is an $\EML$-formula,  so  $\varphi_{\M^L,w}^{\ell}$ is in $\EML$ too.

An analogous argument applies to the characteristic $\EPGML$-formula  
$
\psi_{\M^L,w}^{\ell}
\ :=\
\psi_{\M^L,w}^0 
\ \wedge\
\bigwedge_{j=1}^m \Diamond^{\geq 1}\,\psi_{\M^L,w_j}^{\ell-1}$. By the inductive hypothesis,
each $\psi^{\ell-1}_{\M^{L},w_j}$ is an $\EPML$ formula, and therefore
$\psi^{\ell}_{\M^{L},w}$ is an $\EPML$ formula.
\end{proof}

Combining the pruning construction with the results in
\Cref{lem:bisimilar-trees} and \Cref{lem:ML-characteristic}, we obtain
the following preservation theorem for $\ML$ formulae under injective homomorphisms.

\begin{theorem}\label{injhomML}
The following are equivalent for any \ML{} formula-$\varphi$ of  depth at most $L$:
\begin{itemize}
\item $\varphi$ is preserved under injective homomorphisms;
\item $\varphi$ is equivalent to an  $\EPML$-formula of depth  at most $L$.
\end{itemize}
\end{theorem}
\begin{proof}

\noindent
\emph{($\Rightarrow$)}  Let $\varphi$ be an-$\ML$ formula of $\ell \leq L$ preserved
under injective homomorphisms. 
Let $\C$ be the class of pointed models satisfying
$\varphi$ and  
let $\T=\{\,\unr^{\ell}(\M,w)\mid(\M,w)\in\C\,\}$.
 Since 
$\C$ is invariant under $\ell$-unravellings, we have
$\T \subseteq \C$.

Let $\preceq$ denote the injective homomorphism relation. By  \Cref{cor:finiteness}, $(\T, \preceq)$ is a wqo and hence admits finitely many $\preceq$-minimal elements up to isomorphism. Let
$\Tmin$ contain one representative for each such minimal isomorphism type, and let $\T'$ consist of the prunings of the models in $\Tmin$. 

By \Cref{lem:bisimilar-trees}, each pruning is 
$L$-bisimilar to the original tree. Since $\varphi$ is an $\ML$-formula of depth $\ell$, the class $\C$ is invariant under
$\ell$-bisimulations, and hence $\T' \subseteq \C$. Moreover, 
$\T'$ is finite because $\Tmin$ is finite.

For each pruned tree $(\mathfrak T,t)\in\T'$ consider its
characteristic $\EPGML$-formula
$\psi_{\mathfrak T,t}^{\ell}$. By construction, these formulae
are positive and, by \Cref{lem:ML-characteristic}, 
each $\psi_{\mathfrak T,t}^{\ell}$ is in fact an $\EPML$-formula. 
Let
\[
\Psi \;=\; \bigvee_{(\mathfrak T,t)\in\mathcal T'} \psi_{\mathfrak T,t}^{\ell} .
\]
Then $\Psi$ is a finite disjunction of $\EPML$-formulae of depth at most $\ell$.
By construction, $\Psi$ holds in each model in $\T'$
and hence in each model in $\Tmin$. 
By \Cref{lem:preserve-injective}, it holds in all models in $\T$ and since $\C$ is invariant under $\ell$-unravellings, it follows that $\Psi$
defines $\C$.

%\smallskip
\noindent
\emph{($\Leftarrow$)} Let $\varphi$ be equivalent to an $\EPML$-formula. Since $\EPML \subseteq \EPGML$, $\varphi$ is also equivalent to an $\EPGML$-formula. By  \Cref{injectiveGML}, it is preserved under injective homomorphisms.
\end{proof}

%%%%%%%%%%%%%%%%%%%%%
\subsection{Preservation Under Homomorphisms}
%%%%%%%%%%%%%%%%%%%%%

We now turn to classes preserved under homomorphisms.
We begin by showing that if such a class is invariant under unravelling, then it is definable by existential-positive modal logic ($\EPML$).
Our proof strategy follows the same general pattern
as in the cases of embeddings and injective homomorphisms.
In particular, we define an appropriate form of characteristic $\EPML$-formulae.

\begin{definition}
Let $(\M,w)$ be a pointed model with $\M = (W,R,V)$, the \emph{characteristic $\EPML$-formula} $\chi_{\M,w}^\ell$ of depth $\ell$  is defined as the formula obtained from 
the characteristic $\EGML$-formula $\varphi_{\M,w}^{\ell}$
(see \Cref{def:characteristic-existential}) by removing all negative literals $\neg p$ and replacing each graded modal operator $\Diamond^{\geq k}$ with the basic modal operator $\Diamond$.
\end{definition}

Characteristic $\EPML$-formulae are related to homomorphisms between unravellings by the following lemma.

\begin{lemma}\label{lem:preserve-homomorphism}
Let $(\M, w)$  and $(\N, v)$ be pointed models and let  $\ell \in \mathbb{N}$. Then,
$(\N,v) \models \chi_{\M,w}^\ell$ 
if and only if $\unr^\ell(\M, w) \preceq \unr^\ell(\N, v)$, where
$\preceq$ denotes the homomorphism relation.
\end{lemma}
\begin{proof}
The proof is by induction on $\ell$.

\noindent
\emph{Base case ($\ell=0$).}
The formula  
$\chi^{0}_{\M,w}$ is a conjunction of propositional
symbols satisfied at $w$. Thus, $(\N,v) \models \chi^{0}_{\M,w}$ if and only if $(\N,v)$ satisfies at least the same 
propositional symbols as $(\M,w)$. Since  $\unr^0(\M,w)$ and $\unr^0(\N,v)$ each consist of a single world, this holds if and only if 
$\unr^0(\M,w) \preceq \unr^0(\N,v)$. 

\noindent
\emph{Inductive step.} 
Assume that the claim holds for some $\ell \geq 0$, and consider $\ell+1$.
We use the same notation as in the proof of \Cref{lem:preserve}.

\noindent
\emph{($\Rightarrow$)}
Let $(\N,v) \models \chi^{\ell+1}_{\M,w}$. Then,
for each equivalence class $C_j$ of successors of $w$ in $\M$,
there exists at least one successor $v_{j}$ of $v$ in $\N$ such that
$(\N,v_{j})\models \chi^\ell_{\M,w_j}$. By the inductive hypothesis, for each $j$,
$\unr^\ell(\M,w_j) \preceq \unr^\ell(\N,v_{j})$. 
The unravelling $\unr^{\ell+1}(\M,w)$ consists of the root
together with subtrees rooted at depth-$1$ nodes corresponding
to the successors of $w$. Since homomorphisms need not be injective,
all depth-$1$ nodes corresponding to successors in the same class
$C_j$ may be mapped to the same depth-$1$ node corresponding to
$v_j$. Mapping the root to the root and using the homomorphisms above on the corresponding subtrees yields
a homomorphism from $\unr^{\ell+1}(\M,w)$ to $ \unr^{\ell+1}(\N,v)$.

\smallskip
\noindent
\emph{($\Leftarrow$)}
Let $f$ be a  homomorphism from  $\unr^{\ell+1}(\M, w)$ to $\unr^{\ell+1}(\N, v)$. Then, $f$ maps
the root to the root, implying $(\N, v) \models \chi_{\M,w}^0$.
Fix an equivalence class $C_j$. Every depth-$1$ node corresponding to a successor $w_j$ of $w$ must be mapped by $f$ to some
depth-$1$ node corresponding to a successor $v_j$ of $v$. Restricting
$f$ to the corresponding subtree yields a homomorphism
$\unr^{\ell}(\M,w_j)\to \unr^{\ell}(\N,v_j)$ and by the
inductive hypothesis we obtain $(\N,v_j)\models \chi^{\ell}_{\M,w_j}$.
Hence, $(\N,v)\models \Diamond \chi^{\ell}_{\M,w_j}$ for each $j$ and therefore $(\N,v)\models \chi^{\ell+1}_{\M,w}$.
\end{proof}

The preceding result allows us to establish the homomorphism preservation theorem for unravelling-invariant classes.

\begin{theorem}\label{th:exists-positive-unravelling-ML}
The following are equivalent for any class  $\C$ of pointed models and any $L\in \mathbb{N}$:
\begin{itemize}
\item $\C$ is preserved under homomorphisms and invariant under $L$-unravelling,
\item $\C$ is definable by an $\EPML$-formula of depth at most $L$.
\end{itemize}
\end{theorem}
\begin{proof}
The proof follows the  same general strategy
as that of \Cref{th:preserve}, but now uses  \Cref{lem:preserve-homomorphism} in place of \Cref{lem:preserve}.

\noindent
\emph{($\Rightarrow$)}
Let
$\T=\{\,\unr^L(\M,w)\mid(\M,w)\in\C \,\}$.
Since $\C$ is invariant under $L$-unravelling, we have
$\T\subseteq\C$.
Let $\Tmin\subseteq\T$ be a set of representatives of the 
homomorphism-minimal elements of $\T$. By  \Cref{cor:finiteness},
 $\Tmin$ is finite (alternatively, finiteness also follows from the fact that there are 
 only finitely many $\ML$-inequivalent formulae of depth at most $L$).
For each 
$(\mathfrak T,v)\in\Tmin$, let $\chi^L_{\mathfrak T, v}$ be its characteristic $\EPML$-formula and define 
$\chi_{\C} = \bigvee_{(\mathfrak T, v) \in \Tmin} \chi^{L}_{\mathfrak T,v}$, which is a finite $\EPML$-formula of depth at most $L$.

By construction, $\chi_{\C}$ holds in every model in $\Tmin$
and by \Cref{lem:preserve-homomorphism} it therefore holds in all models in $\T$. Since $\C$ is invariant under $L$-unravelling, it follows that $\chi_{\C}$ defines $\C$.

\noindent
\emph{($\Leftarrow$)}
Suppose that $\C$ is definable by an $\EPML$-formula $\varphi$ of depth at most $L$. Since $\varphi$ is positive, the
standard translation  yields an existential-positive first-order logic formula.
Such formulae are preserved under homomorphisms \cite{DBLP:journals/jacm/Rossman08}.
Moreover, since $\varphi$ has depth at most $L$, it is invariant under $L$-unravelling. 
\end{proof}

We note that \Cref{th:exists-positive-unravelling-ML} yields,
as immediate corollaries, equirank homomorphism
preservation theorems for modal logics. Results of this form were previously obtained
by \citet{DBLP:journals/apal/AbramskyR24} in a general category-theoretic framework. In our setting, they follow directly
from unravelling invariance. Indeed,
every $\GML$ (and $\ML$) formula of depth at most $L$ is invariant under
$L$-unravelling and hence \Cref{th:exists-positive-unravelling-ML} applies to the class of its models.

\begin{corollary}[\cite{DBLP:journals/apal/AbramskyR24}, Theorem 5.20]
The following are equivalent for any $\GML$-formula $\varphi$ of  depth at most $L$:
\begin{itemize}
\item $\varphi$ is preserved under homomorphisms;
\item $\varphi$ is equivalent to an $\EPML$-formula of depth at most $L$.
\end{itemize}
\end{corollary}

%%%%%%%
\subsection{Recovering Rosen's Embedding Preservation Theorem for $\ML$}
%%%%%%

We conclude the preservation results by observing that our framework also
recovers the classical embedding preservation theorem for \ML{}
in the finite setting due to  \citet{DBLP:journals/jolli/Rosen97}. 
This yields an alternative proof based on well-quasi-orders and minimal models, rather than Ehrenfeucht-Fraïssé games.

\begin{theorem}[\cite{DBLP:journals/jolli/Rosen97},  Proposition 6]\label{embeddingML}
The following are equivalent for any \ML{}-formula $\varphi$ of depth at most $L$:
\begin{itemize}
\item $\varphi$ is preserved under embeddings,
\item $\varphi$ is equivalent to an $\EML$-formula of depth  at most $L$.
\end{itemize}
\end{theorem}
\begin{proof}

\noindent
$(\Leftarrow)$ Assume that $\varphi$ is equivalent to an $\EML$-formula of depth  at most $L$. Since
$\EML \subseteq \EGML$, it follows from \Cref{th:preserve} that $\varphi$ is preserved under embeddings.

\noindent
$(\Rightarrow)$ Let $\varphi$ be preserved under embeddings, and
let $\C$ be the class of pointed models of $\varphi$.
Since $\varphi \in \ML$, the class $\C$ is invariant 
under $L$-unravellings. 
By \Cref{th:preserve}, $\C$ is definable by a finite disjunction of characteristic formulae
$\bigvee_{(\mathfrak T,t)\in \Tmin}\varphi^L_{\mathfrak T,t}$, where
$\Tmin$ consists of $\preceq$-minimal
tree-shaped models of height $\leq L$, for $\preceq$ the embedding relation.
 Let $\T'$ consist of the pruning of each tree model in $\Tmin$.
 By \Cref{lem:bisimilar-trees},  each 
$(\mathfrak T,t)\in\Tmin$ is $L$-bisimilar to  its pruning
and hence satisfies the same $\ML$-formulae of depth $\leq L$.
In particular, $\T'\subseteq \C$ and  $\T'$ is finite, since pruning retains at most one 
representative of each $\ML_k$-type for $k\le L$.
By \Cref{lem:ML-characteristic}, formula $\psi = \bigvee_{(\mathfrak{T},w) \in \T'} \varphi^L_{\mathfrak T, w}$ is an $\EML$-formula.
The same argument as in the proof of \Cref{th:preserve} shows that $\psi$ defines $\C$.
\end{proof}

%%%%%%%%%%%%%%
\section{Preservation Theorems for Graph Neural Networks}\label{sec:gnn}
%%%%%%%%%%%%%%%%%

In this section we turn our attention to graph neural networks (GNNs)---machine learning models for graph-structured data.
We focus on one of the flagship classes of GNNs, namely aggregate-combine GNNs \cite{DBLP:conf/iclr/BarceloKM0RS20}.
When such a trained GNN $\GN$ is used for binary classification of nodes in a graph, it defines (analogously to a logical formula) a class of pointed models  that it accepts.

Classes of pointed models definable by GNNs with $L$ layers are preserved under $L$-unravelling \cite{DBLP:conf/iclr/XuHLJ19,DBLP:conf/aaai/0001RFHLRG19,DBLP:conf/iclr/BarceloKM0RS20}.
There are also  striking connections between GNNs and modal logics \cite{DBLP:conf/iclr/BarceloKM0RS20,DBLP:conf/ijcai/NunnSST24,benedikt_et_al:LIPIcs.ICALP.2024.127,boundedGNNs}. 
In particular, existing results show that a first-order property is definable by an aggregate-combine GNN if and only if
it is definable in $\GML$ \cite{DBLP:conf/iclr/BarceloKM0RS20}.
Consequently, our results on unravelling-invariant classes apply directly to GNNs yielding the following preservation theorem.

\begin{corollary}
Let $\GN$ be an aggregate-combine GNN with $L$ layers and let $\C$ be the class of pointed models defined by $\GN$. Then the following 
equivalences hold:
\begin{itemize}
\item $\C$ is preserved under embeddings if and only if $\C$ is definable by an $\EGML$-formula of depth at most $L$;
\item $\C$ is preserved under injective homomorphisms if and only if $\C$ is definable by an $\EPGML$-formula of depth at most $L$;
\item $\C$ is preserved under homomorphisms if and only if $\C$ is definable by an $\EPML$-formula of depth at most $L$.
\end{itemize}
\end{corollary}

In the following subsections, we show that our preservation theorems can be used to give exact logical characterisations of
the expressive power of two classes of GNNs recently studied in the literature \cite{tena2025expressive,DBLP:conf/kr/CucalaG24,DBLP:conf/kr/CucalaGMK23,morris2025sound}.
In particular, we show that monotonic GNNs have exactly the expressive power of $\EPGML$ (\Cref{MONMGNN})
while monotonic GNNs with MAX aggregation have exactly the expressive power of $\EPML$ (\Cref{MAXMGNN}).

%%%
\subsection{Graph Neural Networks}

We consider finite, undirected, simple, and node-labelled graphs $\G = (V, E, \emb)$, where $V$ is a finite set of nodes, $E$ a set of undirected edges without self-loops, and $\emb:V \to \{0,1 \}^d$ assigns to each node a 
binary vector of a fixed dimension $d$ (the \emph{dimension} of $\G$). 
A \emph{pointed graph} is a pair $(\G,v)$ consisting of a graph and a distinguished node.
Pointed graphs of dimension $d$ correspond to Kripke models over a propositional signature of size $d$ in the standard way:
nodes correspond to worlds, each undirected edge $\{u,v\}$ induces accessibility in both directions, and binary vectors correspond to
valuations, where $i$th vector position at a node $v$ determines the truth value of the $i$th proposition at $v$.
Under this correspondence, the notions of embedding and homomorphism (and preservation thereof) extend seamlessly from Kripke models
to pointed  graphs.

A \emph{node classifier}
is a function  mapping
pointed graphs  to  $\true$ (accepted) or $\false$ (rejected). 
Two classifiers are \emph{equivalent} if they compute the same function. 
Two families  $\mathcal{F}$ and $\mathcal{F}'$ of classifiers have the same expressive power, written  $\mathcal{F} \equiv \mathcal{F}'$, if each classifier in $\mathcal{F}$ (respectively, $\mathcal{F}'$) has an equivalent classifier in
$\mathcal{F}'$ (respectively, $\mathcal{F}$).
We consider GNN node classifiers with
 \emph{aggregate-combine} layers \cite{benedikt_et_al:LIPIcs.ICALP.2024.127,DBLP:conf/iclr/BarceloKM0RS20}.
A layer is a pair $( \agg, \comb )$, where  $\agg$ is an \emph{aggregation function} mapping multisets of vectors into single vectors and $\comb$ is  a \emph{combination function} mapping  vectors to  vectors.
Applying a layer to $\G = ( V, E, \emb )$ yields a graph $\G' = ( V, E, \emb' )$ with the same nodes and edges, but with an updated labelling function 
$\emb'$.
For each node $v \in V$, the updated label $\emb'(v)$ is defined by 
\begin{equation}
\comb \Big( \emb(v), \agg( \lBrace  \emb(w) \rBrace_{w \in N_G(v)} )  \Big),  \label{eq:ac}
\end{equation}
where $N_\G(v) = \{w \mid \{v,w \} \in E \}$ denotes the
neighbours of $v$.

Each  aggregation and combination function has a fixed input and output dimension.
For the application of a layer to be well defined, these dimensions must be compatible:
if  $\agg$  has input dimension  $d$ and output dimension $d'$, then $\comb$ has input dimension  $d+d'$.
A \emph{GNN classifier} $\GN$ of dimension $d$ consists of $L$ aggregate-combine layers and a 
classification function $\cls$ mapping vectors to truth values.
The input dimension of  the first layer  is $d$, so that the output dimension of layer $\ell$ matches the input dimension of layer $\ell+1$. 
We write $\lambda^{(\ell)}(v)$ for the vector associated with node $v$ after applying layer $\ell$, with  $\lambda^{(0)}(v)=\emb(v)$ and $\lambda^{(L)}(v)$ the final representation.
The output of $\GN$ on a pointed graph $(\G,v)$ is then defined as
the truth value  $\GN(\G,v) = \cls\bigl(\lambda^{(L)}(v)\bigr)$.

%%%%%%%%%%
\subsection{Expressive Power of Monotonic GNNs}
%%%%%%%%%%

For vectors $\mathbf{x},\mathbf{x}' \in \mathbb{R}^d$ of the same dimension, we write
$\mathbf{x} \leq \mathbf{x}'$ if $x_i \leq x_i'$ for all $i \in [d]$.
For multisets $M$ and $M'$ of vectors of the same dimension, we write $M \leq M'$  if there exists an injection $f \colon M \to M'$ such that $\mathbf{x} \leq f(\mathbf{x})$, for every $\mathbf{x} \in M$.
For example,
\[
\lBrace [-1,2], [0,1] \rBrace \;\leq\; \lBrace [-1,2], [0,3], [-1,-1] \rBrace,
\]
as witnessed by the injection $f$ defined as $f([-1,2])=[-1,2]$ and $f([0,1])=[0,3]$.

A function $f$ mapping vectors or multisets to vectors or numbers---such as combination, aggregation, or classification---is \emph{non-decreasing} if $x \leq y$ implies $f(x) \leq f(y)$ for all admissible arguments $x,y$.

A GNN is \emph{monotonic} if all its aggregation functions, combination functions, and its
final classification function are non-decreasing. 

A \emph{positive-weight GNN} is a GNN  where each layer $\ell$ has the form:
$$
ReLU \Big(
\mathbf{b}_{\ell} +
     \emb(v) \mathbf{A}_{\ell}  +
         \agg_{\ell} \big( \lBrace  \emb(w) \mid w \in N_G(v) \rBrace \big) \mathbf{C}_{\ell}  \Big)
$$
where $\mathbf{b}_\ell\in\mathbb{R}^{d_\ell}$, $\mathbf{A}_\ell,\mathbf{C}_\ell\in\mathbb{R}_{\ge 0}^{d_{\ell-1}\times 
d_\ell}$,
$\agg_{\ell}$ is one of $\mathrm{SUM}$, $\mathrm{MAX}$, or $\mathrm{max}\text{-}k\text{-}\mathrm{SUM}$,
and the classification function $\cls$ is a threshold function of the form
$\cls(\mathbf{x})=1$ iff $x_i\ge t$ (or $x_i>t$) for all $i$, for some  threshold $t\in\mathbb{R}$.

\begin{proposition}\label{prop:mon-gnns}
Positive-weight GNNs using MAX, SUM, or max-k-sum aggregation are monotonic. In contrast, positive-weight GNNs using
MEAN aggregation are not monotonic. 
\end{proposition}
\begin{proof}
In all these GNNs, the combination function is monotonic, since the ReLU activation function
is monotonic and, for each layer $\ell$, all entries of the weight matrices $\mathbf{A}_{\ell}$ and $\mathbf{C}_{\ell}$ are non-negative.
The aggregation functions MAX, SUM, and max-k-sum are also monotonic with respect to the multiset order.
In contrast, MEAN is not monotonic. Indeed,
$\lBrace 1 \rBrace \leq \lBrace 1 ,0 \rBrace$, but 
$MEAN(\lBrace 1 \rBrace) =1 \not\leq 0.5 =
MEAN(\lBrace 1, 0 \rBrace)$.
\end{proof}

\begin{theorem}\label{th:mon-GNN-preservation}
Monotonic GNNs are preserved under injective homomorphisms, but not under arbitrary homomorphisms.
\end{theorem}
\begin{proof}
Let $h \colon (G,v)\to(G',v')$ be an injective homomorphism between pointed graphs, and let $\GN$ be a monotonic GNN such that $\GN(G,v)=1$.
We show that $\GN(G',v')=1$.

Since the classification function $\cls$ is non-decreasing, it suffices to show that $\lambda^{(L)}(v) \leq \lambda^{(L)}(v')$, for $L$ the number of layers of $\GN$. 
We prove by induction on $\ell \leq L$ that $\lambda^{(\ell)}(w) \leq \lambda^{(\ell)}(h(w))$, for every node $w$ in $G$. 

The base case $(\ell = 0)$ holds because $h$ preserves node labels. For the inductive step, assume the claim 
holds for $\ell < L$. Since $h$ is an injective homomorphism, it induces an injection from the multiset of 
neighbours of $w$ in $G$ to the multiset of neighbours of $h(w)$ in $G'$. By the induction hypothesis, this
yields $\lBrace \lambda^{(\ell)}(u)\mid u\in N_G(w)\rBrace \leq \lBrace \lambda^{(\ell)}(u')\mid u'\in N_{G'}(h(w))\rBrace$.
By monotonicity of the aggregation and combination functions, we conclude  that $\lambda^{(\ell+1)}(w)\le \lambda^{(\ell+1)}(h(w))$.
Hence, $\GN(G',v')=1$.

To see that monotonic GNNs are not preserved under arbitrary homomorphisms, let $G=(V,E,\emb)$ with $V=\{v,u,w\}$, $E = \{ \{v,u\}, \{v,w \} \}$, $\emb(v)=\emb(u)=\emb(w)=1$, and let  $G'=(V',E',\emb')$ with $V'=\{v',u' \}$, $E= \{ \{v',u' \} \}$, and $\emb'(v)=\emb'(u)=1$.
Define $\GN$ with a single layer, aggregation function SUM, parameters $\mathbf{b}=0$ and $\mathbf{A}=\mathbf{C}=1$, and classification function
$\cls(x)=1$ iff $x \geq 3$.
Then,  $\GN(G,v)=1$, but $\GN(G',v')=0$.
However, the function $f$  defined by $f(v)=v'$, $f(u)=u'$ and $f(w) = u'$ is a (non-injective) homomorphism from $(G,v)$ to $(G',v')$.
\end{proof}

The above subsumes the result that monotonic max-$k$-sum GNNs are closed under injective homomorphisms \cite[Proposition 7]{tena2025expressive}.

A \emph{monotonic-MAX GNN} is a monotonic GNN  using only MAX as aggregation function.
The following theorem isolates MAX aggregation as a case in which injectivity is no longer required

\begin{theorem}\label{th:mon-GNN-max-preservation}
Monotonic-MAX GNNs are preserved under arbitrary homomorphisms.
\end{theorem}

\begin{proof}
The proof follows the same general structure as that of \Cref{th:mon-GNN-preservation}.
Let $h \colon (G,v)\to(G',v')$ be a homomorphism between pointed graphs, and let $\GN$ be a monotonic-MAX GNN with $\GN(G,v)=1$.
We show that $\GN(G',v')=1$.
As before, since the classification function $\cls$ is non-decreasing, it suffices to show that
$\lambda^{(L)}(v) \leq \lambda^{(L)}(v')$, where $L$ is the number of layers in $\GN$. 

To this end, we prove by induction on $\ell \leq L$ that $\lambda^{(\ell)}(w) \leq \lambda^{(\ell)}(h(w))$, for each node $w$ in $G$. 
The base case $(\ell = 0)$ holds because $h$ preserves node labels. For the inductive step, assume  the claim 
holds for $\ell < L$. Since $h$ is a homomorphism, each neighbour $u$ of $w$ in $G$ is mapped to a neighbour $h(u)$ of $h(w)$ in $G'$. By the induction hypothesis, $\lambda^{(\ell)}(u) \leq \lambda^{(\ell)}(h(u))$ and hence 
 $MAX(\lBrace \lambda^{(\ell)}(u)\mid u\in N_G(w)\rBrace) \leq MAX(\lBrace \lambda^{(\ell)}(u')\mid u'\in N_{G'}(h(w))\rBrace)$. By  monotonicity of the combination function, we conclude that  $\lambda^{(\ell+1)}(w)\le \lambda^{(\ell+1)}(h(w))$, completing the induction.
Hence $\GN(G',v')=1$.
\end{proof}

\begin{theorem}\label{MONMGNN}
Monotonic GNNs have exactly the same expressive power as $\EPGML$.
\end{theorem}

\begin{proof}
We first show that each monotonic GNN $\GN$ admits an equivalent $\EPGML$ formula.
Let $\C$ be the class of pointed graphs accepted by $\GN$.
By Theorem~\ref{th:mon-GNN-preservation}, $\C$ is preserved under injective homomorphisms.
Moreover, $\C$ is invariant under $L$-unravellings, where $L$ is the number of layers of $\GN$, since after $L$ layers the representation of a node depends only on its radius-$L$ neighbourhood \cite{DBLP:conf/iclr/BarceloKM0RS20}.
Viewing graphs as Kripke models, and using that $\C$ is preserved under injective homomorphisms and invariant under $L$-unravellings, \Cref{th:exists-positive-unravelling} implies that $\C$ is definable by an $\EPGML$-formula.

Conversely, let $\varphi$ be an $\EPGML$-formula.
By adapting the construction in Proposition~4.1 of \cite{DBLP:conf/iclr/BarceloKM0RS20} to the existential-positive 
fragment, we obtain a GNN that evaluates $\varphi$.
Since $\varphi$ is positive, the construction in   \cite{DBLP:conf/iclr/BarceloKM0RS20} uses only non-negative weights, SUM aggregation, and ReLU activations.
Hence the resulting GNN is monotonic.
\end{proof}

\begin{theorem}\label{MAXMGNN}
Monotonic-MAX GNNs  have exactly the same expressive power as $\EPML$.
\end{theorem}

\begin{proof}
Let $\GN$ be a monotonic GNN in which each aggregation function is MAX, and let $\C$ be the class of pointed graphs
accepted by $\GN$. By Theorem \ref{th:mon-GNN-max-preservation}, the class $\C$ is preserved under arbitrary homomorphisms.
Moreover, as for any monotonic GNN, $\C$ is invariant under $L$-unravellings, where $L$ is the number of layers of $\GN$. Hence, by Theorem \ref{th:exists-positive-unravelling-ML},
the class $\C$ is definable by an $\EPML$-formula.
Conversely, let $\varphi$ be an $\EPML$-formula. Theorem 6 in \cite{boundedGNNs}, shows that 
$\varphi$ is equivalent to a positive-weight GNN using only MAX as aggregation. By \Cref{prop:mon-gnns},  such GNN is monotonic.
\end{proof}

\section{Conclusions}

In this paper, we developed a uniform approach to preservation theorems for classes of models that are invariant under unravelling, 
covering both (graded) modal logics and  graph neural networks.
Our results establish a correspondence between semantic preservation properties and logical definability:
preservation under embeddings coincides with definability in existential graded modal logic,
preservation under injective homomorphisms with definability in existential-positive graded modal logic, and
preservation under homomorphisms with definability in existential-positive modal logic.
Technically, our approach relies on a structural well-quasi-ordering result for tree-shaped models of uniform bounded height.
This framework allows us to obtain new preservation theorems for graded modal logic, to recover and subsume classical results for basic modal 
logic over finite structures, and to provide alternative proofs avoiding compactness and Ehrenfeucht-Fraïssé games.
Furthermore, our results on GNNs place recent expressive power analyses of GNNs in a unified logical framework using elegant techniques
from finite model theory.

Our work opens several directions for future research.
On the logical side, it is natural to investigate whether similar preservation theorems can be established for extensions of modal logic, such as
logics with nominals, global modalities, or fixed-point operators, as well as for restricted fragments, including Horn or core variants.
On the algorithmic and machine-learning side, an important open question is whether our proof strategy—based on unravelling invariance, well-quasi-orders, and minimal models—can be adapted to analyse the expressive power of other classes of graph neural networks, such as architectures
with attention, normalization, or higher-order message passing.
We believe that further interaction between finite model theory and the theory of GNNs will continue to yield fruitful insights in both areas.

%\begin{acks}
%Anonymous
%\end{acks}

\bibliographystyle{ACM-Reference-Format}
\bibliography{biblio}

\end{document}